\documentclass[sigconf,authorversion=false]{acmart}
\usepackage{graphicx}
\usepackage{hyperref}
\usepackage{multicol}

\usepackage{booktabs}

\usepackage{times}
\usepackage{subfigure}
\usepackage{mathtools}
\newtheorem{fact}{Fact}

\newcommand{\para}[1]{\smallskip \noindent {\bf #1}}
\newcommand{\parit}[1]{\smallskip \noindent {\em #1}}
\newcommand{\eat}[1]{}
\newcommand{\eeps}{\ensuremath{e^{\epsilon}}}

\newcommand{\var}{\mathrm{Var}} 
\newcommand{\cov}{\mathrm{Cov}}

\DeclareMathOperator{\E}{\mathrm{\mathbb{E}}}
\newcommand{\frack}[2]{#1/#2}

\setcopyright{none} 

\author{Graham Cormode}
\authornote{}
\affiliation{%
  \institution{University Of Warwick, UK}
  } %
    \email{g.cormode@warwick.ac.uk}

 \author{Tejas Kulkarni}
\authornote{}
 \affiliation{%
  \institution{University Of Warwick, UK}
}
   \email{tejasvijaykulkarni@Gmail.com}

\author{Divesh Srivastava}
\affiliation{%
  \institution{AT\&T Labs-Research, USA}
  }
  \email{divesh@research.att.com}

\begin{document}
\title{Answering Range Queries Under Local Differential Privacy} 

\begin{abstract}
Counting the fraction of a population having an input within a
specified interval i.e. a \emph{range query}, is a fundamental data
analysis primitive.
Range queries can also be used to compute other interesting statistics
such as \emph{quantiles}, and to build prediction models.
However, frequently the data is subject to privacy concerns when it
is drawn from individuals, and relates for example to their financial, health,
religious or political status. 
In this paper, we introduce and analyze methods to support range
queries under the local variant of differential
privacy~\cite{SLNRS:2011}, an emerging standard for privacy-preserving
data analysis.


The local model requires that each user releases a noisy view of her
private data under a privacy guarantee. 
While many works address the problem of range queries in the
trusted aggregator setting, this problem has not been
addressed specifically under untrusted aggregation (local DP) model
even though many primitives have been developed recently for 
estimating a discrete distribution.
We describe and analyze two classes of approaches for range queries,
based on hierarchical histograms and the Haar wavelet transform.
We show that both have strong theoretical accuracy guarantees on
variance.
In practice, both methods are fast and require minimal computation and
communication resources.
Our experiments show that the wavelet approach is most accurate in high privacy settings, while
the hierarchical approach dominates for weaker privacy requirements. 
\end{abstract}




\maketitle

\section{Introduction}

All data analysis fundamentally depends on a basic understanding of how the data is
distributed.
Many sophisticated data analysis and machine learning techniques are
built on top of primitives that describe where data points are
located, or what is the data density in a given region.
That is, we need to provide accurate answers to estimates of the data
density at a given point or within a range.
Consequently, we need to ensure that such queries can be answered
accurately under a variety of data access models.

This remains the case when the data is sensitive, comprised of the
personal details of many individuals.
Here, we still need to answer range queries accurately, but also meet
high standards of privacy, typically by ensuring that answers are
subject to sufficient bounded perturbations that each individual's
data is protected.
In this work, we adopt the recently popular model of Local
Differential Privacy (LDP).
Under LDP, individuals retain control of their own private data,
by revealing only randomized transformations of their input.
Aggregating the reports of sufficiently many users gives accurate
answers, and allows complex analysis and models to be built, while
preserving each individual's privacy.

LDP has risen to prominence in recent years due to its adoption and
widespread deployment by major technology companies, including Google~\cite{rappor:15},
Apple~\cite{applewhitepaper:17} and Microsoft~\cite{telemetry:17}.
These applications rely at their heart on allowing frequency
estimation within large data domains (e.g. the space of all words, or
of all URLs).
Consequently, strong locally private solutions are known for this point estimation
problem.
It is therefore surprising to us that no prior work has explicitly
addressed the question of {\em range queries} under LDP.
Range queries are perhaps of wider application than point queries,
from their inherent utility to describe data, through their immediate
uses to address cumulative distribution and quantile queries, up to
their ability to instantiate classification and regression models for
description and prediction.

In this paper, we tackle the question of how to define protocols to
answer range queries under strict LDP guarantees.
Our main focus throughout is on one-dimensional discrete domains,
which provides substantial technical challenges under the strict model
of LDP.
These ideas naturally adapt to multiple dimensions, as we
discuss briefly as an extension. 
 A first approach to answer range queries is to simply pose each
point query that constitutes the range.
This works tolerably well for short ranges over small domains, but
rapidly degenerates for larger inputs.
Instead, we adapt ideas from computational geometry, and show how
hierarchical and wavelet decompositions can be used to reduce the
error.
This approach is suggested by prior work in the centralized privacy
model, but we find some important differences, and reach different
conclusions about the optimal way to include data and set parameters
in the local model.
In particular, we see that approaches based on hierarchical
decomposition and wavelet transformations are both effective and offer
similar accuracy for this problem; whereas, naive approaches that
directly evaluate range queries via point estimates are inaccurate and
frequently unwieldy.  

\para{Our contributions.}
In more detail, our contributions are as follows:
We provide background on the model of Local Differential Privacy (LDP)
and
related efforts for range queries in Section~\ref{sec:prior}.
Then in Section~\ref{sec:prelims}, we summarize the approaches to
answering point queries under LDP, which are a building block for our
approaches. 
Our core conceptual contribution (Section~\ref{sec:range}) comes from proposing and analyzing
several different approaches to answering one-dimensional range
queries. 

\para{$\bullet$}
We first formalize the problem and show that the simple approach of
summing a sequence of point queries entails error (measured as
variance) that grows linearly with the length of the range (Section~\ref{sec:flat}).  

\para{$\bullet$}
In Section~\ref{sec:hh}, we consider hierarchical approaches,
generalizing the idea of a binary tree.
We show that the variance grows only logarithmically with the length
of the range.
Post-processing of the noisy observations can remove inconsistencies,
and reduces the constants in the variance, allowing an optimal
braching factor for the tree to be determined.

\para{$\bullet$}
  The last approach is based on the Discrete Haar wavelet transform
  (DHT, described in Section~\ref{sec:dht}).
  Here the variance is bounded in terms of the logarithm of the domain
  size, and no post-processing is needed.
  The variance bound is similar but not directly comparable to that in the
  hierarchical approach. 

\medskip
Once we have a general method to answer range queries, we can apply it
to the special case of prefix queries, and to find order statistics
(medians and quantiles).
We perform an empirical comparison of our methods in
Section~\ref{sec:expts}.
Our conclusion is that both the hierarchical and DHT approach are
effective for domains of moderate size and upwards.
The accuracy is very good when there is a large population of users
contributing their (noisy) data.
Further, the related costs (computational resources required by each
user and the data aggregator, and the amount of information sent by
each user) are very low for these methods, making them practical to
deploy at scale. We show that the wavelet approach is most accurate in high privacy settings, while the hierarchical approach dominates for weaker privacy requirements.
We conclude by considering extensions of our scenario, such as
multidimensional data (Section~\ref{sec:conclusions}).


\section{Related Work}
\label{sec:prior}

\para{Range queries.}
Primitives to support range queries are necessary in a variety of data
processing scenarios. 
Exact range queries can be answered by simply scanning the data and
counting the number of tuples that fall within the range; faster
answers are possible by pre-processing, such as sorting the data (for
one-dimensional ranges).
Multi-dimensional range queries are addressed by geometric data
structures such as $k$-$d$ trees or quadtrees~\cite{Samet:05}.
As the dimension increases, these methods suffer from the ``curse of
dimensionality'', and it is usually faster to simply scan the data.

Various approaches exist to approximately answer range queries.
A random sample of the data allows the answer on the sample to be
extrapolated; to give an answer with an additive $\epsilon$ guarantee
requires a sample of size $O(\frac{1}{\epsilon^{2}})$~\cite{Cormode:Garofalakis:Haas:Jermaine:12}.
Other data structures, based on histograms or streaming data sketches
can answer one-dimensional range queries with the same accuracy
guarantee and with a space cost of
$O(1/\epsilon)$~\cite{Cormode:Garofalakis:Haas:Jermaine:12}.
However, these methods do not naturally translate to the private
setting, since they retain information about a subset of the input
tuples exactly, which tends to conflict with formal statistical
privacy guarantees. 

\para{Local Differential Privacy (LDP).}
The model of local differential privacy has risen in popularity in
recent years in theory and in practice as a special case of
differential privacy.
It has long been observed that local data perturbation methods,
epitomized by Randomized Response~\cite{Warner:65} also meet the
definition of Differential Privacy~\cite{Dwork:Roth:14}.
However, in the last few years, the model of local data perturbation
has risen in prominence: initially from a theoretical
interest~\cite{d:13}, but subsequently from a practical
perspective~\cite{rappor:15}.
A substantial amount of effort has been put into the question of
collecting simple popularity statistics, by scaling randomized
response to handle a larger domain of
possibilities~\cite{applewhitepaper:17,telemetry:17,Bassily:Nissim:Stemmer:Thakurta:17,Wang:Blocki:Li:Jha:17}.
The current state of the art solutions involve a combination of ideas
from data transformation, sketching and hash projections to reduce the
communication cost for each user, and computational effort for the
data aggregator to put the information
together~\cite{Bassily:Nissim:Stemmer:Thakurta:17,Wang:Blocki:Li:Jha:17}.

Building on this, there has been substantial effort to solve a variety
of problems in the local model, including:
language modeling and text prediction~\cite{Chang:Thakurta:18};
higher order and marginal statistics~\cite{calm:2018,rappor2:16,Cormode:Kulkarni:Srivastava:18}; 
social network and graph modeling~\cite{GLCZ:2018,Qin:Yu:Yang:Khalil:Xiao:Ren:17};
and various machine learning, recommendation and model building tasks~\cite{WXYZHSSY:2019,d:13,Xiao:16,Zheng:Mou:Wang:17,SKSX:2018}
However, among this collection of work, we are not aware of any work
that directly or indirectly addresses the question of allowing range
queries to be answered in the strict local model, where no interaction
is allowed between users and aggregator.

\para{Private Range queries.}
In the centralized model of privacy, there has been more extensive
consideration of range queries.
Part of our contribution is to show how several of these ideas can
be translated to the local model, then to provide customized analysis
for the resulting algorithms.
Much early work on differentially private histograms considered range
queries as a natural target~\cite{Dwork:McSherry:Nissim:Smith:06,Hardt:Ligett:McSherry:12}. 
However, simply summing up histogram entries leads to large errors for
long range queries. 

Xiao {\em et al.}~\cite{dht:11} considered adding
noise in the Haar wavelet domain, while 
Hay {\em et al.}~\cite{HRMS:10} formalized the approach of keeping a
hierarchical representation of data.
Both approaches promise error that scales only logarithmically with
the length of the range. 
These results were refined by
Qardaji {\em et al.}~\cite{QYL:13}, who compared the two approaches
and optimized parameter settings.
The conclusion there was that a hierarchical approach with moderate
fan-out (of 16) was preferable, more than halving the error from the
Haar approach.
A parallel line of work considered two-dimensional range queries, introducing the
notion of private spatial decompositions based on $k$-$d$ trees and
quadtrees~\cite{Cormode:Procopiuc:Srivastava:Shen:Yu:12}.
Subsequent work argued that shallow hierarchical structures were
often preferable, with only a few levels of refinement~\cite{Qardaji:Yang:Li:12}.

\section{Model And Preliminaries}
\label{sec:prelims}
\subsection{Local Differential Privacy}
\label{sec:ldp}
Initial work on differential privacy assumed the presence of a
\emph{trusted aggregator}, who curates all the private information of
individuals, and releases information through a perturbation
algorithm. In practice, individuals may be reluctant to share private
information with a data aggregator. The \textit{local} variant of
differential privacy instead captures the case when each user $i$ only
has their local view of the dataset $S$ (typically, they only know
their own data point $z_i$) and she independently releases information
about her input through an instance of a DP algorithm.
This model has received widespread industrial adoption, including by
Google \cite{rappor:15,rappor2:16}, Apple \cite{applewhitepaper:17},
Microsoft \cite{telemetry:17} and Snap \cite{snap:18} for tasks like
heavy hitter identification (e.g., most used emojis), training word
prediction models, anomaly detection, and measuring app usage.

In the simplest setting, we assume each participant $i \in [N]$ has an
input $z_i$ drawn from some global discrete or continuous distribution
$\theta$ over a domain $\mathcal{Z}$.
We do not assume that users share any trust relationship with each
other, and so do not communicate amongst themselves. 
Implicitly, there is also an (untrusted) aggregator interested in estimating some statistics over the private dataset $\{z_i\}_{i=1}^N$. 

\para{Formal definition of Local Differential Privacy (LDP) \cite{SLNRS:2011}.}
A randomized function $F$ is $\epsilon$-locally differentially private if for all possible pairs of  $z_i,z_i'\sim Z$ and for every possible output tuple $O$ in the range of $F$:
\[
\Pr[F(z_i)=O] \leq e^{\epsilon} \Pr[F(z_{i}')=O].
\]
This is a local instantiation of differential
privacy~\cite{Dwork:Roth:14}, where the perturbation mechanism $F$ is applied to
each data point independently.
In contrast to the centralized model, perturbation under LDP happens at the user's end.

\subsection{Point Queries and Frequency Oracles}
\label{sec:hist}

A basic question in the LDP model is to answer {\em point queries} on
the distribution: to estimate the frequency of any given element $z$
from the domain $\mathcal{Z}$. 
Answering such queries form the underpinning for a variety of
applications such as population surveys, machine learning, spatial
analysis and, as we shall see, our objective of quantiles and range
queries. 


In the point query problem, each $i$ holds a private item $z_{i}$
drawn from a public set $\mathcal{Z}, |\mathcal{Z}|=D$ using an
unknown common discrete distribution $\theta$.
That is, $\theta_z$ is the probability that a randomly sampled input
element is equal to $z \in \mathcal{Z}$. 
The goal is to provide a protocol in the LDP model (i.e. steps that
each user and the aggregator should follow) so the aggregator can
estimate $\theta$ as $\widehat{\theta}$ as accurately as possible. 
Solutions for this problem are referred to as providing a {\em
  frequency oracle}. 

Several variant constructions of frequency oracles have been described in
recent years. 
In each case, 
the users perturb their input locally via tools such as linear
transformation and random sampling, and send the result to the
aggregator.
These noisy reports are aggregated and an appropriate bias correction
is applied to them to reconstruct the frequency for each item in $\mathcal{Z}$.
The error in estimation is generally quantified by the \emph{mean
  squared error}.
We know that the mean squared error can be decomposed into
\emph{(squared) bias} and \emph{variance}.
Often estimators for these mechanisms are \emph{unbiased} and have the
same variance $V_F$ for all items in the input domain.
Hence, the variance can be used interchangably with squared error,
after scaling. 
The mechanisms vary based on their computation and communication
costs, and the accuracy (variance) obtained.
In most cases, the variance is proportional to 
$\frac{1}{N(e^\epsilon-1)^{2}}$.


\para{Optimized Unary Encoding (OUE).}
A classical approach to releasing a single bit of data with a
privacy guarantee is Randomized Response (RR), due to Wagner
\cite{Warner:65}.
Here, we either report the true value of the input or its complement
with appropriately chosen probabilities. 
To generalize to inputs from larger domains,  we represent $v_i$
as the sparse binary vector $e_{v_i}$ (where
$e_j[j]=1$ and 0 elsewhere), and randomly flip each bit of $e_{v_i}$ to
obtain the (non-sparse) binary vector $o_i$.
Naively, this corresponds to applying  one-bit randomized response~\cite{Warner:65} to
each bit independently.
Wang et al.\ \cite{Wang:Blocki:Li:Jha:17} proposed a variant of this
scheme that reduces the variance for larger $D$. 

\parit {Perturbation:} Each user $i$ flips each bit at each location $j
\in [D]$ of $e_i$ using the following distribution.
\begin{align*}
\Pr[o_i[j]=1]=
\begin{cases}
\frac{1}{2},      \text{ if     } e_i[j]=1\\
\frac{1}  {1+e^{\epsilon}},    \text{ if      } e_i[j]=0\\
\end{cases}
\end{align*}
Finally user $i$ sends the perturbed input $o_i$ to the aggregator. 	


\parit {Aggregation:}
$\widehat{\theta}[z] = 
\left(\frac{\sum_{i=1}^{N} o_i[z]}{N}+\frac{N}{e^{\epsilon}+1}\right)\big/\left({\frac{1}{2}-\frac{1}{e^{\epsilon}+1}}\right)$

\parit {Variance:} $V_F=\frac{4e^{\epsilon}}{N(e^{\epsilon}-1)^{2}}$

OUE
does not scale well to very large $D$ due to large
communication complexity (i.e., $D$ bits from each user), and the
consequent computation cost for the user ($O(D)$ time to flip the
bits).
Subsequent mechanisms have smaller communication cost than OUE.

\eat{\para{Optimized Unary Signed Encoding (OUSE).}
Describe the signed version of OUE for Haar inputs $\{-1,0,+1\}$ .
Or a real-valued version that applies to any value in the range
$[-1,1]$. 
}
\eat{
\para{Generalized Randomized Response (GRR) \cite{KBR:16}}.
For small domains, it makes sense to try to report the user's input
directly, with a false report chosen with appropriate probability to
ensure privacy. 

\parit{{Perturbation:}}
With probability $p=\frac{e^{\epsilon}}{e^{\epsilon}+D-1}$, we report the true
$z_{i}$.
Otherwise, we report $j \neq z_i$ sampled u.a.r. (uniformly at random)
so that each value is chosen
with probability $\frac{1}{e^\epsilon + D - 1}$. 

\parit{{Aggregation:}} $\widehat{\theta}[z]=\frac{(D-1)\frac{\sum_{i=1}^{N}
    I_{r_{i}=z}}{N}+p-1}{pD-1}$.

\parit{{Variance:}} $V_F=\frac{D-2+e^{\epsilon}}{N(e^{\epsilon}-1)^{2}}$.

Since the value of $p$ reduces with $D$, GRR has a large variance for
items from  large domains.
Nevertheless, it is a preferred mechanism when domain size $D$ is small.
 \cite{Wang:Blocki:Li:Jha:17} showed that GRR achieves the least variance when $D < 3e^{\epsilon}+2$.
}

\para{Optimal Local Hashing (OLH) \cite{Wang:Blocki:Li:Jha:17}.}
The OLH method aims to reduce the impact of dimensionality on
accuracy by employing  \textit{universal hash functions}\footnote{A
  family of hash functions $\mathbb{H}=\{H:[D] \rightarrow [g]\}$ is
  said to be universal  if $\forall z_i,z_j \in [D] , z_i \neq z_j : \Pr_{H
    \in \mathbb{H}}[H(z_i) = H(z_j)] \leq \frac{1}{g}$ i.e. collision
  probability behaves uniformly.}.
More specifically, each user samples a hash function $H:  [D]
\rightarrow [g]$ ($g \ll D$) u.a.r from a universal family $\mathbb{H}$ and perturbs the hashed input.

\parit{Perturbation:} User $i$ samples a $H_i \in \mathbb{H}$ u.a.r and
computes $h_i=H_i(v_i)$.
User $i$ then perturbs $h_i \in [g]$ using a version of RR generalized
for categorical inputs~\cite{KBR:16}.
Specifically, each user reports $H_i$ and, 
with probability $p=\frac{e^{\epsilon}}{e^{\epsilon}+g-1}$
gives the true $h_i$, else 
she reports a value sampled u.a.r from $[g]$. 

\parit{Aggregation:} The aggregator collects the perturbed hash values from all users.
For each hash value $h_i$, the aggregator computes a frequency vector
for all items in the original domain, based on which items would produce the hash
value $h_i$ under $H_i$.
All $N$ such histograms are added together to give $T \in \mathbb{R}^{D}$ and an unbiased estimator for each frequency for all elements in the original domain is given by the correction $\widehat{\theta}[j]=\frac{T[j]-\frac{N}{g}}{p-\frac{1}{g}}$.  

\parit{Variance:}
Setting $g=e^{\epsilon}+1$ minimizes the variance to be
$V_F =\frac{4p(1-p)}{N(2p-1)^{2}}=\frac{4e^{\epsilon}}{N(e^{\epsilon}-1)^{2}}$. 
OLH has the same variance as OUE and it more economical on
communication.
However, a major downside is that it is compute intensive in terms of
the decoding time at the aggregator's side, which is prohibitive for
very large dimensions (say, for $D$ above tens of thousands), since the time cost is proportional to
$O(ND)$. 

\begin{figure}
\centering
  \centering
  \scriptsize
\[\frac{1}{\sqrt{8}}\begin{pmatrix*}[r]
 1 &  1 &  1 &  1 &  1 &  1 &  1 &  1 \\  
 1 & -1 &  1 & -1 &  1 & -1 &  1 & -1 \\
 1 &  1 & -1 & -1 &  1 &  1 & -1 & -1 \\
 1 & -1 & -1 &  1 &  1 & -1 & -1 &  1 \\
 1 &  1 &  1 &  1 & -1 & -1 & -1 & -1 \\
 1 & -1 &  1 & -1 & -1 &  1 & -1 &  1 \\
 1 &  1 & -1 & -1 & -1 & -1 &  1 &  1 \\
 1 & -1 & -1 & -1 & -1 &  1 &  1 & -1 
 \end{pmatrix*}\]

\caption{Hadamard Transform Matrix for $D=8$.}
\label{fig:hadamard}
\end{figure}

\para{Hadamard Randomized Response (HRR) \cite{Cormode:Kulkarni:Srivastava:18,Xiao:16}.}
The Discrete Fourier (Hadamard) transform is described by an orthogonal, symmetric matrix $\phi$ of dimension $D \times D$ (where $D$ is a power of 2).
Each entry in $\phi$ is 

\centerline{$\textstyle
\phi[i][j]=\frac{1}{\sqrt{D}}(-1)^{\langle i,j\rangle},$}
where $\langle i,j\rangle$ is the number of $1$'s that $i$ and $j$ agree on
in their binary representation.
The (full) Hadamard transformation (HT) of user's input $v_i$ is the $v_i$th
column of $\phi$ i.e. $\phi
\times e_i$.
For convenience, the user can scale $\phi$ up by $\sqrt{D}$ to give values either
$-1$ or $1$.
Figure~\ref{fig:hadamard} shows an example of the Hadamard transformation matrix.


\parit{Perturbation:} User $i$
samples an index $j\in [D]$ u.a.r and
perturbs $\phi[v_i][j]\in\{-1,1\}$ with binary randomized response,
keeping the value with probability $p$, and flipping it with
probability $1-p$.
Finally user $i$ releases the perturbed coefficient $o_j$
and $j$.

\parit{Aggregation:}
Consider each report from each user.
With probability $p$, the report is the true value of the coefficient;
with probability $1-p$, we receive its negation.
Hence, we should divide the reported value by $2p-1$ to obtain an
unbiased random variable whose expectation is the correct value. 
The aggregator can then compute the observed sum of each perturbed 
coefficient $j$ as $O_j$.
An unbiased estimation of the $j$th Hadamard coefficient
$\widehat{c_j}$ (with the $\frac{1}{\sqrt{D}}$ factor restored) is
given by 
$\widehat{c_j} 
=\frac{O_j}{\sqrt{D}(2p-1)}$. 
Therefore, the aggregator can compute an unbiased estimator for each
coefficient, and then apply the inverse
transform to produce $\widehat{\theta}$.

\parit{Variance:} The variance of each user report is given by the squared error of our unbiased estimator.
With probability $p$, the squared error is $(1 - \frac{1}{2p-1})^2/D$,
else the squared error is $(1+\frac{1}{2p-1})^2/D$.
Then, we can expand the variance for each report as
\begin{align*}
\textstyle
  \frac{p(2p-2)^2 + (1-p)4p^2}{D(2p-1)^2}
  = \frac{4p(1-p)^2 + 4p^2(1-p)}{D(2p-1)^2} =
 \frac{4p(1-p)}{D(2p-1)^2}
   \end{align*}

\eat{
  We can find the variance in the estimation of any $j$th Hadamard coefficient $\widehat{c_j}$ as
\begin{align*}
& \var[\E[\widehat{c_j}]] =\frac{\var[\E[O_j]]}{D(2p-1)^2}= \frac{\var[-4pc_j-1+2p+2c_j]}{D(2p-1)^2} \\ &= \frac{-4c_j(4p(1-p))+8p(1-p))}{D(1-2p)^{2}}  =\frac{8p(1-p)(2c_j-1)}{D(2p-1)^{2}} \\ &\leq \frac{8p(1-p)}{D(2p-1)^{2}}
\end{align*}
}

There are $N$ total reports, each of which samples one of $D$
coefficients at random.
Observing that the estimate of any frequency in the original domain is
a linear combination of Hadamard coefficients with unit Euclidean
norm,
we can find an expression for the value of $V_F$ as
$V_F = \frac{4p(1-p)}{\frac{N}{D}D(2p-1)^2} = \frac{4p(1-p)}{N(2p-1)^2}$.
Using $p=\frac{e^\epsilon}{1+e^\epsilon}$ (to ensure LDP), we find
$V_F = \frac{4e^\epsilon}{N(e^\epsilon-1)^2}$.

This method achieves a good compromise between accuracy and
communication since each user transmits only $\lceil\log_2 D\rceil + 1$ bits to
describe the index $j$ and the perturbed coefficient, respectively.
Also, the aggregator can reconstruct the frequencies in the original domain
by computing the estimated coefficients and then inverting HT with
$O(N + D \log D)$ operations, versus $O(ND)$ for OLH.

\medskip
Thus, we have three representative mechanisms to implement a frequency
oracle. Each one provides $\epsilon$-LDP, by considering the probability of seeing the same output from the user if her input were to change.
There are other frequency oracles mechanisms developed offering
similar or weaker variance bounds (e.g. \cite{rappor2:16,telemetry:17}) and resouce trade-offs but we do not include them for brevity.  

\section{Range Queries}
\label{sec:range}

\subsection{Problem Definition}

We next formally define the range queries that we would like to
support. 
As in Section~\ref{sec:hist}, we assume $N$
non-colluding individuals each with a private item $z_i \in [D]$.
For any $a < b,  a \in [D], b \in [D]$, a range query $R_{[a,b]} \geq 0$ is
to compute
 \begin{align*} 
R_{[a,b]}= \frac{1}{N} \sum_{i=1}^{N} \mathbb{I}_{ a \leq  z_{i} \leq b} 
 \end{align*}	
where $\mathbb{I}_{p}$ is a binary variable that takes the value $1$ if the predicate $p$ is true and $0$ otherwise. 

 \begin{definition}(Range Query Release Problem)
 Given a set of $N$ users, the goal is to collect information guaranteeing $\epsilon$-LDP to allow approximation of any closed interval of length $r \in[1,D]$. Let $\widehat{R}$ be an estimation of interval $R$ of length $r$ computed using a mechanism $F$. Then the quality of $F$ is measured by the squared error $(\widehat{R}-R)^{2}$.
 \end{definition}

\begin{figure}
\centering
\subfigure[Dyadic decomposition of the domain showing internal node weights.]{
  \includegraphics[width=0.45\textwidth]{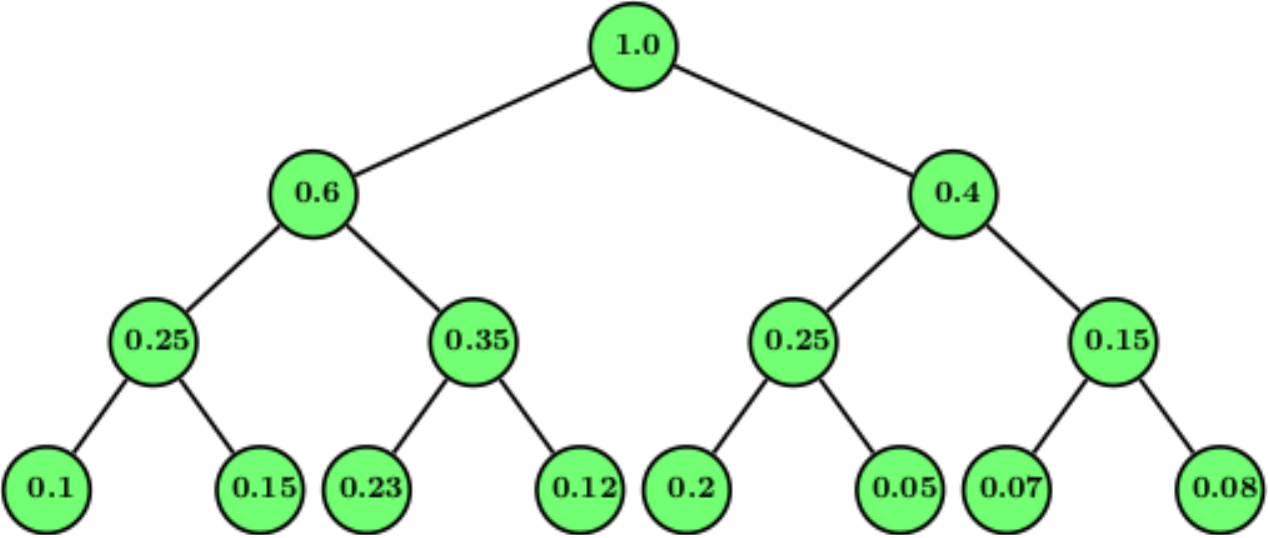}
  \label{fig:dd-weights}
}
\subfigure[Local views for two users $i$ and $j$ ($z_i=1$ and $z_j=5$) with corresponding root to leaf paths marked.]
{\includegraphics[width=0.25\textwidth]{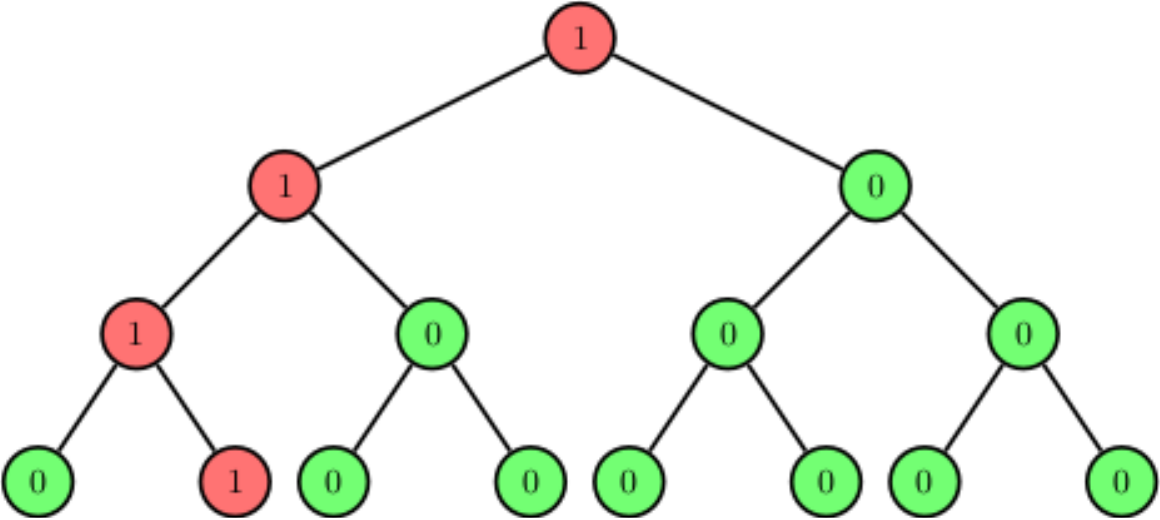}
\includegraphics[width=0.25\textwidth]{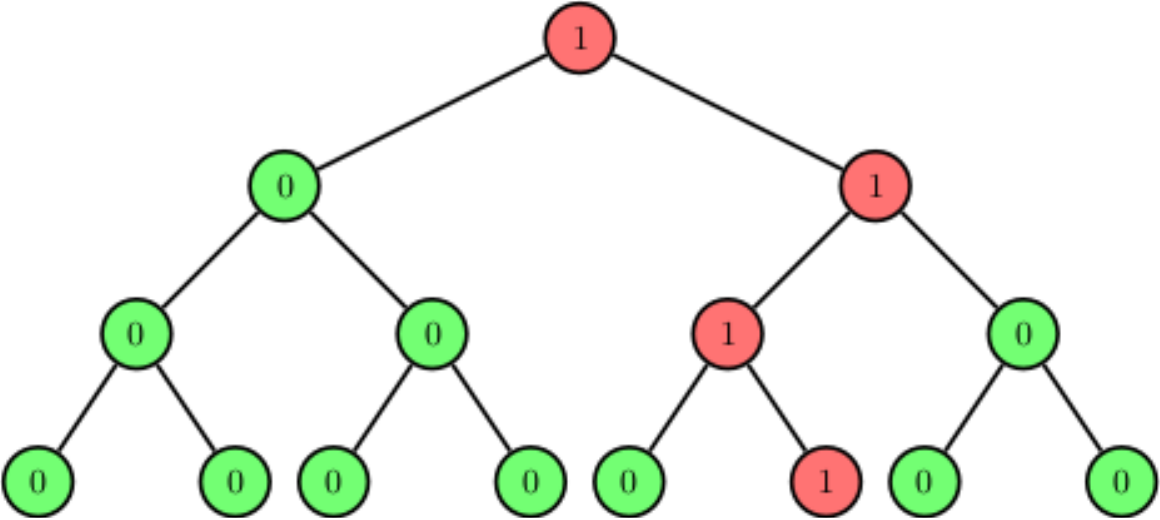}
\label{fig:dd-path}
}
\caption{
An example for dyadic decomposition ($B=2$)}
\label{fig:dd}
\end{figure}


\subsection{Flat Solutions}
\label{sec:flat}

One can observe that for an interval
$[a,b]$, $R_{[a,b]}=\sum_{i=a}^{b}f_i$, where $f_i$ is the
(fractional) frequency of the item $i \in[D]$.
Therefore a first approach is to simply sum up estimated frequencies
for every item in the range, where estimates are provided by an
$\epsilon$-LDP frequency oracle: 
$\widehat{R}_{[a,b]}=\sum_{i=a}^{b}\widehat{\theta}_i$.
We denote this approach (instantiated by a choice of frequency oracle $F$)
as \emph{flat} algorithms.

\begin{fact}
For any range query $R$ of length $r$ answered using a flat method
with frequency oracle $F$, $\var[\widehat{R}-R]= rV_F$ 
\end{fact}

Note that the variance grows linearly with the interval size which can be as large as $DV_{F}$.
 
 \begin{lemma} 
 The average worst case squared error over evaluation of ${D \choose 2}$ queries $\mathcal{E}$ is $\frac{1}{3}(D+2)V_F.$
 \end{lemma}
 \begin{proof}
 There are $D-r+1$ queries of length $r$. Hence the average error is $\mathcal{E} = {\sum_{r=1}^{D} r(D-r+1)V_F}\big/{{D \choose 2}}=\frac{1}{3}(D+2)V_F$
 \end{proof}
 
 \subsection{Hierarchical Solutions}
\label{sec:hh}
We can view the problem of answering range queries in terms of
representing the frequency distribution via some collection of histograms,
and producing the estimate by combining information from bins
in the histograms.
The ``flat'' approach instantiates this, and keeps one bin for each individual
element.
This is necessary in order to answer range queries of length 1
(i.e. point queries).
However, as observed above, if we have access only to point queries,
then the error grows in proportion to the length of the
range.
It is therefore natural to keep additional bins over subranges of the
data. 
A classical approach is to impose a hierarchy on the domain
items in such a way that the frequency of each item contributes to
multiple bins of varying granularity.
With such structure in place, we can answer a given query by adding counts from a relatively small number of bins.  
There are many hierarchical methods possible to compute histograms. 
Several of these have been tried in the context of centralized DP. 
But to the best of our knowledge, the methods that work best in
centralized DP tend to rely on a complete view on the distribution, or
would require multiple interactions between users and aggregator when
translated to the local model. 
This motivates us to choose more simple yet effective strategies
for histogram construction in the LDP setting. 
We start with the standard notion of $B$-adic intervals and a useful
property of $B$-adic decompositions. 

\begin{fact}
For  $j \in [\log_B D]$ and $B \in \mathbb{N}^{+}$, an interval is B-adic if it is of the form $kB^{j}...(k+1)B^{j}-1$ i.e. its length is of a power of $B$ and starts with an integer multiple of its length.
\end{fact}


\begin{fact} Any sub-range of length $[a,b]$ of length $r$ from $[D]$ can be decomposed into $\leq (B-1)(2\log_{B} r+1)$ disjoint $B$-adic ranges.   
\end{fact}

For example, for $D=32,B=2$, the interval $[2, 22]$ can be decomposed into sub-intervals
$[2,3]\cup [4,7] \cup [8,15] \cup [16,19] \cup [20,21] \cup [22,22]$.

The $B$-adic decomposition can be understood as organizing the domain
under a complete $B$-ary tree where each node corresponds to a bin of
a unique $B$-adic range.
The root holds the entire range and the leaves hold the counts for
unit sized intervals.
A range query can be answered by a walk over the tree similar to the standard \emph{pre-order traversal} and therefore a range query can be answered with at most 
$2(B-1)(1+\log_{B} r)$ nodes, and at most $2(B-1)(\log_{B}
D+\frac12)-1$ in the worst case. 

\subsection{ Hierarchical Histograms (HH)}
Now we describe our framework for computing hierarchical histograms.
All algorithms follow a similar structure but differ on the perturbation primitive $F$ they use:

\para{Input transformation}: user $i$ locally arranges the input $z_i
\in [D]$ in the form of a full $B$-ary tree of height $h$.
Then $z_i$ defines a unique path from a leaf to the root with
a weight of 1  attached to each node on the path, and zero
elsewhere.
Figure~\ref{fig:dd} shows an example.
Figure~\ref{fig:dd-weights} shows the dyadic ($B=2$) decomposition of the
input vector $[0.1,0.15,0.23,0.12,0.2,0.05,0.07,0.08]$, where the weights on internal nodes
are the sum of the weights in their subtree. 
Figure~\ref{fig:dd-path} illustrates two user's local views ($z_i=1$ and $z_j=5$). In each local histogram, the nodes in the path from leaf to the root are shaded in red and have a weight of 1 on each node.

\para{Perturbation}: $i$ {\em samples} a level $l \in [h]$ with probability
$p_{l}$.
There are $2^l$ nodes at this level, with exactly one node of weight
one and the rest zero.
Hence, we can apply one of the mechanisms from Section~\ref{sec:prelims}. 
User $i$ perturbs this vector using some frequency oracle $F$ and sends the perturbed information to the aggregator along with the level id $l$.

\para{Aggregation}: The aggregator builds an empty tree with the same
dimensions and adds the (unbiased) contribution from each user to the
corresponding nodes, to estimate the fraction of the input at each
node.  Range queries are answered by aggregating the nodes from the
$B$-adic decomposition of the range.

\para{Key difference from the centralized case}:
Hierarchical histograms have been proposed and evaluated in the
centralized case.
However, the key difference here comes from how we generate
information about each level.
In the centralized case, the norm is to split the ``error budget''
$\epsilon$ into $h$ pieces, and report the {\em count} of users in
each node; in contrast, we have each user sample a single level,
and the aggregator estimates the {\em fraction} of users in each
node.
The reason for sampling instead of splitting emerges from the
analysis: splitting would lead to an error proportional to $h^2$,
whereas sampling gives an error which is at most proportional to $h$.
Because sampling introduces some variation into the number of users
reporting at each level, we work in terms of fractions rather than
counts; this is important for the subsequent post-processing step. 


\smallskip
In summary, the approach of hierarchical decomposition extends to LDP
by observing the fact that it is a \emph{linear} transformation of the
original input domain.
This means that adding information from the hierarchical
decomposition of each individual's input yields the decomposition of
the entire population.
Next we evaluate the error in estimation using the hierarchical methods.
 
\para{Error behavior for Hierarchical Histograms.}
We begin by showing that the overall variance can be expressed in
terms of the variance of the frequency oracle used, $V_F$. 
In what follows, we denote hierarchical histograms aggregated with fanout B as $\text{HH}_B$.

\begin{theorem}
When answering a range query of length $r$ using a
primitive $F$, the worst case variance $V_r$ under the $\text{HH}_B$ framework is
${V}_r \leq V_F\sum_{l=1}^{\alpha} (2B-1)\frac{1}{p_l}$ where $\alpha=
(\lceil\log_B r\rceil +1)$.
\label{thm:hhvar}
\end{theorem} 
\begin{proof}
Recall that all the methods we consider have the same (asymptotic)
variance bound $V_F = O\big(\frac{e^{\epsilon}}{N(e^{\epsilon} - 1)^2}\big)$,
  with $N$ denoting the number of users contributing to the
  mechanism.
Importantly, this does not depend on the domain size $D$, and so we
can write $V_F \le \psi_{F}(\epsilon)/N$, where $\psi_{F}(\epsilon)$ is a
constant for method $F$ that depends on $\epsilon$. 
This means that once we fix the method $F$, the variance $V_l$ for any
node  at level $l$ will be the same, and is determined by $N_l$, the
number of users reporting on level $l$.
The range query $R_{[a,b]}$ of length $r$ is decomposed into
at $2(B-1)$ nodes at each level, for
$\alpha = \lceil\log_B r\rceil + 1$ levels (from leaves upwards). 
So we can bound the total variance $\mathcal{V}^{r}$ in our estimate
by
\[\sum_{l=1}^{\alpha} (2B-1)V_l = \sum_{l=1}^{\alpha} (2B-1)V_F/p_l
= (2B-1)V_F \sum_{l=1}^{\alpha} \frac{1}{p_l}
\]
using the fact that (in expectation) $N_l = p_l N$. 
\end{proof}

In the worst case, $\alpha = h$, and we can minimize this bound by a
uniform level sampling procedure: 

\begin{lemma}
The quantity $\sum_{l=1}^{h}\frac{1}{p_l}$ subject to 
$0 \le p_l \le 1$ and $\sum_{l=1}^{h}p_l=1$ 
is minimized by setting $p_l=\frac{1}{h}$.
\label{lemma:minsumprob}
\end{lemma}

\begin{proof}
We use the Lagrange multiplier technique, and 
define a new function $\mathcal{L}$, introducing a new variable $\lambda$.
\begin{align*}
  \textstyle
\mathcal{L}(p_1,..,p_h,\lambda)=(\sum_{l=1}^{h}\frac{1}{p_l}) + \lambda(\sum
_{l=1}^{h}p_l-1)
\end{align*}
We compute the gradient as 
\[\begin{array}{rl}
  \textstyle
\nabla_{p_1,..,p_h,\lambda} \mathcal{L}(p_1,..,p_h,\lambda) &=\Big(\frac{	\partial \mathcal{L}}{\partial p_1},\frac{	\partial \mathcal{L}}{\partial p_2},..,\frac{	\partial \mathcal{L}}{\partial p_h},\frac{	\partial \mathcal{L}}{\partial \lambda} \Big)\\ &= \Big( \lambda-\frac{1}{p^{2}_1},..,\lambda-\frac{1}{p^{2}_h},\sum
_{l=1}^{h}\frac{1}{p_l}-1  \Big)
\end{array}
\]
Equating $\nabla_{p_1,..,p_h,\lambda} \mathcal{L}(p_1,..,p_h,\lambda)=0$, we get $\lambda=\frac{1}{p^{2}_1}=\frac{1}{p^{2}_2}=...=\frac{1}{p^{2}_h}$ and $\sum
_{l=1}^{h}\frac{1}{p_l}=1$.
Hence, 
$p_l=1/\sqrt{\lambda}=1/h$.
\end{proof}
Then, setting $p_l = \frac{1}{h}$ in Theorem~\ref{thm:hhvar}
gives
\begin{equation}
  V_r \le (2B-1)V_F h(\lceil \log_B r \rceil + 1). 
  \label{eq:hhvar}
  \end{equation}
  
\para{Hierarchical versus flat methods.}
The benefit of the HH approach over the baseline flat method depends
on the factor $(2B-1)h\alpha$ versus the quantity $r$.
Note that $h= \log_B D + O(1)$ and $\alpha = \log_B r + O(1)$,
so we obtain an improvement over flat methods
when $r > 2B\log^2_B D$, for example.
When $D$ is very small, this may not be achieved:
for $D=64$ and $B=2$, this condition yields $r > 128 > D$.
However, for larger $D$, say $D=2^{16}$ and $B=2$, we obtain
$r > 1024$, which corresponds to approximately 1.5\% of the range. 



 
\begin{theorem}
The worst case average (squared) error incurred  while answering all
${ D \choose 2}$ range queries using HH$_B$, $\mathcal{E}_B$, is
(approximately) $2(B-1)V_F \log_B D \log_B \Big(\frac{ 3D^{2}}{1+2D}
\Big)$
\label{thm:worst_err}
\end{theorem}

\begin{proof}
   We obtain the bound by summing over all range lengths $r$.
   For a given length $r$, there are $D-r+1$ possible ranges. 
Hence,
 \begin{align*}
   \mathcal{E}_B  & \leq \frac{\sum_{r=1}^{D} V_{r}(D-r+1)}{\frack{D(D-1)}{2}} \\ &= \frac{(2(B-1)V_F \log_{B} D)\sum_{r=1}^{D} \log_B r (D-r+1)}{\frack{D(D-1)}{2}}   \\ &= \frac{2(B-1)V_F \log_B D \Big[ (D+1) \log_B (\prod_{r=1}^{D} r ) - \sum_{r=1}^{D}\log_B r^{r} \Big] }{\frack{D(D-1)}{2}} 
 \end{align*}
We find bounds on each of the two components separately.

\smallskip
1. Using Stirling's approximation 
  we have \\
  $\log_B D! \leq \log_B (D^{(D+\frac{1}{2})}e^{1-D}) < (D+1)\log_B D$.

\smallskip
2. Writing $P=\sum_{r=1}^{D} r =\frack{D(D+1)}{2}$ and $Q =
  \sum_{r=1}^{D} r^{2} = \frack{D(D+1)(2D+1)}{6}$, we make use of
  Jensen's inequality to get
\begin{align*}
\sum_{r=1}^{D} r\log_B r &= P\sum_{r=1}^{D} \frac{r}{P} \log_B r \leq P \log_B ( \sum_{r=1}^{D} r\frac{r}{P})  \\ &= 
P \log_B (\frack{Q}{P})=\frack{D(D+1)}{2} \log_B \Big(\frack{1+2D}{3}\Big)
\end{align*}
 Plugging these upper bounds in to the main expression, 
\begin{align*}
\mathcal{E}_B &
\textstyle  < \frac{2(B-1)V_F \log_B D\Big[(D+1)^{2}\log_B D -   \frac{D(D+1)}{2} \log_B \Big(\frac{1+2D}{3}  \Big)\Big]}{\frack{D(D-1)}{2}} \\
&=
\textstyle
2(B-1)V_F\log_B D \Big[\frac{2(D+1)^{2}\log_B D}{D(D-1)} -\frac{D+1}{D-1}\log_B \Big(\frac{1+2D}{3}\Big) \Big] \\ & \approx 2(B-1)V_F \log_B D \log_B \Big(\frac{ 3D^{2}}{1+2D} \Big) \text{as } D\to\infty. 
\end{align*}
\end{proof}

\para{Key difference from the centralized case}:
Similar looking bounds are known in centralized case, for example due to
Qardaji et al.~\cite{QYL:13}, but with some key differences.
There, the bound (simplified) is proportional to $(B-1)h^3 V_F$ rather
than the $(B-1)h^2 V_F$ we see here.
Note however that in the centralized case $V_F$ scales proportionate
to $1/N^2$ rather than $1/N$ in the local case: a necessary cost to
provide local privacy guarantees.


\para{Optimal branching factor for HH$_B$.} In general, increasing the fan-out
has two consequences under our algorithmic framework. Large $B$
reduces the tree height, which increases accuracy of estimation per
node since larger population is allocated to each level.
But this also means that we can  require more nodes at each level to evaluate a query
which tends to increase the total error incurred during evaluation.
We would like to find a branching factor that balances these two
effects.
We use the expression for the variance in \eqref{eq:hhvar}
to find the
optimal branching factor for a given $D$.
We first compute the gradient of the function $2(B-1)\log_B(r)\log_B(D)$.
Differentiating w.r.t. $B$ we get
\[\begin{array}{rl}
  \nabla&= \frac{D}{dB}\Big[\frac{2(B-1)\ln(D)\ln(r)}{\ln^{2}B}\Big] =
  2\ln D \ln r \frac{D}{dB} \Big[\frac{B-1}{\ln^{2}B}  \Big]
  \\ &= \frac{2\ln(D)\ln(r)}{(\ln^{2} B)^{2}} \big( \ln^{2} B\frac{D}{dB}[B-1]-
      (B-1)\frac{D}{dB}[\ln^{2}B] \big) 
  \\ &= \frack{2\ln D \ln r \big( \ln^{2}B - \frac{2}{B}(B-1)\ln B  \big)}{\ln^{4} B} \\&= \frack{2\ln D \ln r(B \ln B -2B +2)}{B \ln^{3} B}
\end{array}
\]
We now seek a $B$ such that the derivative $\nabla=0$.
The numerical solution is (approximately) $B=4.922$.
Hence we minimize the variance by choosing $B$ to be 4 or 5. 
This is again in contrast to the centralized case, where the optimal
branching factor is determined to be approximately $16$~\cite{QYL:13}.

\subsection{Post-processing for consistency} 
\label{sec:consistency}
There is some redundancy in the information
materialized by the HH approach: we obtain estimates for the weight of
each internal
node, as well as its child nodes, which should sum to the parent weight.
Hence, the accuracy of the HH framework can be further improved 
by finding the \emph{least squares} solution for the weight of each node
taking into account all the information we have about it,
i.e.\ for each node $v$, we approximate the (fractional) frequency $f(v)$ with
$\widehat{f}(v)$ such that $|| f(v)-\widehat{f}(v)||_2$ is minimized subject
to the consistency constraints.
We can invoke the Gauss-Markov theorem since the variance of all our
estimates are equal, and hence the least squares solution is the
best linear unbiased estimator.
\begin{lemma}
The least-squares estimated counts reduce the associated variance by a
factor of at least $\frac{B}{B+1}$ in a hierarchy of fan-out $B$. 
  \label{lem:linearalg}
\end{lemma}
\begin{proof}
We begin by considering the linear algebraic formulation. 
Let $H$ denote the $n \times D$ matrix that encodes the hierarchy,
where $n$ is the number of nodes in the tree structure.
For instance, if we consider a single level tree with $B$ leaves, then
$H = \left[ \begin{array}{c} \mathbf{1}_D \\ I_D \end{array}\right]$, where
$\mathbf{1}_D$ is the $D$-length vector of all 1s, and
$I_D$ is the $D\times D$ identity matrix. 
Let $\mathbf{x}$ denote the vector of reconstructed (noisy) frequencies of nodes.
 Then the optimal least-squares estimate of the true counts can be
 written as $\widehat{\mathbf{c}} = (H^TH)^{-1}H^T\mathbf{x}$. 
Denote a range query $R_{[a,b]}$ as the length $D$ vector that is 1
for indices between $a$ and $b$, and 0 otherwise.
Then the answer to our range query is
$R_{[a,b]}^T \widehat{\mathbf{f}}$.
  The variance associated with  query $R_{[a,b]}$ can be found as
\begin{align*}
  \var[R_{[a,b]}^T \widehat{\mathbf{c}}] &=
  \var[R_{[a,b]}^T (H^TH)^{-1}H^T\mathbf{x}] \\&=
  R_{[a,b]}^T(H^TH)^{-1}H^T \cov(\mathbf{x}) H((H^TH)^{-1})^T  R_{[a,b]}
    \\&=
  R_{[a,b]}^T  (H^TH)^{-1}H^T V_F I_D H((H^TH)^{-1})^T) R_{[a,b]}
  \\&= V_F R_{[a,b]}^T  (H^TH)^{-1}(H^T H)((H^TH)^{-1})^T) R_{[a,b]}
  \\&=V_F R_{[a,b]}^T (H^TH)^{-1} R_{[a,b]}
\end{align*}

First, consider the simple case when $H$ is a single level tree
with $B$ leaves.
Then we have
$H^T H = \mathbf{1}_{B\times B} + I_B$, where $\mathbf{1}_{B \times
  B}$ denotes the $B \times B$ matrix of all ones.
We can verify that
$(H^TH)^{-1} = ((B+1)I_B - \mathbf{1}_{B\times B})/(B+1)$.
From this we can quickly read off the variance of any range query.
For a point query, the associated variance is simply $B/(B+1)V_F$, 
while for a query of length $r$, the variance is
$(rB - r(r-1))/(B+1)V_F$. 
Observe that the variance for the whole range $r=B$ is just
$B/(B+1)V_F$, and that the maximum variance is for a range of just
under half the length, $r=(B+1)/2$, which gives a bound of

\centerline{$V_F(B+1)(B+1)/(4(B+1)) = (B+1)V_F/4$.}

The same approach can be used for hierarchies with more than one
level.
However, while there is considerable structure to be studied here,
there is no simple closed form, and forming $(H^TH)^{-1}$ can be
inconvenient for large $D$.
Instead, for each level, we can apply the argument above
between the noisy counts for any node and its $B$ children.
This shows that if we applied this estimation procedure to just these
counts, we would obtain a bound of $B/(B+1)V_F$ to any node (parent or
child), and at most $(B+1)V_F/4$ for any sum of node counts.
Therefore, if we find the optimal least squares estimates, their
(minimal) variance can be at most this much. 
\end{proof}

Consequently, after this constrained inference, the error variance at each node is at
most $\frac{BV_F}{B+1}$. 
It is possible to give a tighter bound for nodes higher up in the
hierarchy: the variance reduces by
$\frac{B^i}{\sum_{j=0}^i B^j}$ for level $i$ (counting up from level
1, the leaves).
This approaches $(B-1)/B$, from above; however, we adopt the
simpler $B/(B+1)$ bound for clarity. 

This modified variance affects the worst case error, and hence our
calculation of an optimal branching factor.
From the above proof, we can obtain a new bound on the worst case
error of
$(B+1)V_F/2$ for every level touched by the query (that is,
$(B+1)V_F/4$ for the left and right fringe of the query).
This equates to
$(B+1)V_F \log_B(r) \log_B(D)/2$ total variance.
Differentiating w.r.t. $B$, we find
\[
\begin{array}{rl}
\nabla &= \frac{d}{dB} \Big[ (B+1)\log_{B}(r)\log_B(D)V_F/2 \Big] \\ &=
\frack{\ln(r)\ln(D) (B\ln B - 2B -2)}{B\ln^3B}
\end{array}\]

\eat{
We apply the factor $\frac{B}{B+1}$ to the worst case bound
$2(B-1)\log_B(r) \log_{B}(D)V_F.$
Now, when we differentiate w.r.t. $B$ 
\begin{align*}
\nabla &= \frac{d}{dB} \Big[ \frac{2(B-1)\log_{B}(r)\log_B(D)BV_F}{B+1} \Big] \\ &=
\frac{2V_F\ln(D)\ln(r)((B^{2}+2B-1)\ln(B)-2B^{2}+2)}{(B+1)^{2}\ln^{3}(B)}
\end{align*}
}

Consequently, the value that minimizes $\nabla$ is $B\approx 9.18$ ---
larger than without consistency. 
This implies a constant factor reduction in the variance in range
queries from post-processing.
Specifically, if we pick $B=8$ (a power of 2), then
this bound on variance is
\begin{equation}
9 V_F\log_2(r) \log_2(D)/(2 \log_2^2 8) = \frac12 V_F\log_2(r)
\log_2(D),
\label{eq:hhvarbound}
\end{equation}
compared to $\frac74 V_F \log_2(r) \log_2(D)/4$ for HH$_4$ without
consistency. 
We confirm this reduction in error experimentally in
Section~\ref{sec:expts}.

We can make use of the structure of the hierarchy to provide a
simple linear-time procedure to compute optimal estimates. 
This approach was introduced in the centralized case by Hay et
al.\ \cite{HRMS:10}.
Their efficient two-stage process can be translated to the local
model. 

\para{Stage 1: Weighted Averaging}: Traversing the tree bottom up, we
use the weighted average of a node's original reconstructed frequency
$f(.)$ and the sum of its children's (adjusted) weights to update the node’s reconstructed weight.
For a non-leaf node $v$, its adjusted weight is a weighted combination
as follows: 
\begin{align*}
  \textstyle
 \bar{f}(v)= \frac{B^{i} - B^{i-1}}{B^{i}-1} f(v)  + \frac{B^{i-1}-1}{B^{i}-1} \sum_{u \in \text{child}(v)} \bar{f}(u) 
\end{align*}

\noindent
{\bf Stage 2: Mean Consistency}:
This step makes sure that for each node, its weight is equal to the
sum of its children's values.
This is done by
dividing the difference between parent's weight and children's total weight equally among children. 
For a non-root node~$v$, 
\begin{align*}
\textstyle
  \widehat{f}(v)=  \bar{f}(v)  + \frac{1}{B}\Big[\bar{f}(p(v)) -\sum_{u \in \text{child}(v)} \bar{f}(u)\Big]
\end{align*} 
where $\bar{f}(p(v))$ is the weight of $v$'s parent after weighted averaging.
The values of $\widehat{f}$ achieve the minimum $L_2$ solution.

Finally, we note that the cost of this post-processing is relatively
low for the aggregator: each of the two steps can be computed in a
linear pass over the tree structure. 
A useful property of finding the least squares solution is that it
enforces the consistency property: the final estimate for each node is
equal to the sum of its children.
Thus, it does not matter how we try to answer a range query (just
adding up leaves, or subtracting some counts from others) --- we will
obtain the same result. 

\para{Key difference from the centralized case.}
Our post-processing is  influenced by a sequence of papers in
the centralized case.
However, we do observe some important points of departure.
First, because users sample levels, we work with the distribution of
frequencies across each level, rather than counts, as the counts are
not guaranteed to sum up exactly.
Secondly, our analysis method allows us to give an upper bound on the
variance at every level in the tree -- prior work gave a mixture of
upper and lower bounds on variances.
This, in conjunction with our bound on covariances allows us to give a
tighter bound on the variance for a range query, and to find a bound
on the optimal branching factor after taking into account the
post-processing, which has not been done previously. 

\begin{figure}[t]
  \centering
  \scriptsize
\[\frac{1}{\sqrt{8}}\begin{pmatrix*}[r]
 1  & 1 &   \sqrt{2}  &  0  &    2 &  0   &   0  &    0  \\  
  1 & 1 &   \sqrt{2}  &  0  &  -2  &  0  &   0  &   0   \\
  1 &  1 &  -\sqrt{2} &  0   &    0  &  2  &   0  &   0   \\
  1 & 1 &   -\sqrt{2} &  0   &   0   & -2  &   0  &  0   \\
  1 &-1 &   0    &  \sqrt{2} &   0   &  0   &  2  &  0   \\
  1 &-1 &   0    &  \sqrt{2} &   0   &  0   & -2  &  0   \\
  1 &-1 &   0    & -\sqrt{2} &   0   &  0   &  0   &   2  \\
  1 &-1 &   0    & -\sqrt{2} &   0   &  0   &  0   &  -2  \\

 \end{pmatrix*}\]
\caption{DHT matrix for $D=8$.}
 \label{fig:haar}
\end{figure}

\subsection{Discrete Haar Transform (DHT)}
\label{sec:dht}

The Discrete Haar Transform (DHT) provides an alternative approach to
summarizing data for the purpose of answering range queries. 
DHT is a popular data synopsis tool that relies on a hierarchical
(binary tree-based) decomposition of the data.
DHT can be understood as performing recursive pairwise averaging and differencing of
our data at different granularities, as opposed to the HH approach
which gathers sums of values. 
The method imposes a full binary tree structure over the domain, where
$h(v)$ is the height of node $v$, counting up from the leaves (level
0). 
The Haar coefficient $c_v$ for a node $v$ is computed as
$c_v=\frac{C_l - C_r}{2^{h(v)/2}}$, where $C_l,C_r$ are the sum of counts of
all leaves in the left and right subtree of $v$.
In the local case when $z_i$ represents a leaf of the tree, 
there is exactly one non-zero haar coefficient at each level $l$ with
value $\pm\frac{1}{2^{l/2}}$.
The DHT can also be represented as a matrix $H_D$ of dimension $D \times D$
(where $D$ is a power of 2) with each row $j$ encoding the Haar
coefficients for item $j \in [D]$. Figure~\ref{fig:haar} shows an
instance of DHT matrix for $D=8$, i.e. $H_8$.
In each row, the zeroth coefficient is always $\frac{1}{\sqrt{D}}$.
The subsequent $D-1$ entries give weights for each of $h=\log_2 D$ levels.
For instance, for row 0, $\{\frac{1}{\sqrt{2}},0,0,0\}$, $\{\frac12,0\}$ and $\{\frac{1}{\sqrt{8}}\}$ are the sets for coefficients for levels 1, 2 and 3. 
We can decode the count at any leaf node $v$ 
by taking the inner product of the vector of Haar coefficients with the row
of $H_D$ corresponding to $v$. 
Observe that we only need $h$ coefficients to answer a point query. 
 
\para{Answering a range query.}
A similar fact holds for range queries. 
We can answer any range query by first summing all rows of $H_D$ that
correspond to leaf nodes within the range, then taking the inner
product of this with the coefficient vector.
We can observe that for an internal node in the binary tree, if it is fully contained (or fully excluded) by the range, then it contributes zero to the sum.
Hence, we only need coefficients corresponding to nodes that are cut
by the range query: there are at most $2h$ of these.
The main benefit of DHT comes from the fact that all coefficients are
independent, and there is no redundant information. 
Therefore we obtain a certain amount of consistency by design: any set
of Haar coefficients uniquely determines an input vector, and there is
no need to apply the post-processing step described in
Section~\ref{sec:consistency}.

%

\para{Our algorithmic framework.}
For convenience, we rescale each coefficient reported by a user
at a non-root node to be from $\{-1,0,1\}$, and apply the scaling
factor later in the procedure. 
Similar to the HH approach, each user samples a level $l$ with
probability $p_l$ and perturbs the coefficients from that level using a
suitable perturbation primitive.
Each user then reports her noisy coefficients along with the level.
The aggregator, after accepting all reports, prepares a similar tree
and applies the correction to make an unbiased estimation of each Haar
coefficient.
The aggregator can evaluate range queries using the (unbiased but still
noisy) coefficients.

\para{Perturbing Haar coefficients.}
As with hierarchical histogram methods, where each level is a sparse
(one hot) vector, there are several choices for how to release
information about the sampled level in the Haar tree.
The only difference is that previously the non-zero entry in the level
was always a 1 value; for Haar, it can be a $-1$ or a $1$.
There are various straightforward ways to adapt the methods that we
have already (see, for example, \cite{sb:15,Xiao:16,telemetry:17}). 
We choose to adapt the Hadamard Randomized Response (HRR) method,
described in Section~\ref{sec:hist}. 
First, this is convenient: it immediately works for negative
valued weights without any modification.
But it also minimizes the communication effort for the users: they
summarize their whole level with a single bit (plus the description of
the level and Hadmard coefficient chosen).
We have confirmed this choice empirically in calibration experiments
(omitted for brevity): HRR is consistent with other choices in terms of
accuracy, and so is preferred for its convenience and compactness. 


Recall that the (scaled) Hadamard transform of a sparse binary vector $e_i$
is equivalent to selecting the $i$th row/column from the Hadamard
matrix.
When we transform $-e_i$, the Hadamard coefficients remain
binary, with their signs negated.
Hence we use HRR for perturbing levelwise Haar coefficients.
At the root level, where there is a single coefficient, this is
equivalent to  1 bit RR.
The 0th coefficient can be hardcoded to $\frac{N}{\sqrt{D}}$ since it
does not require perturbation. 
We refer to this algorithm as HaarHRR.

\para{Error behavior for HaarHRR.}
As mentioned before, we answer an arbitrary query of length $r$ by
taking a weighted combination of at most $2h$ coefficients.
A coefficient $u$ at level $l$ contributes to the answer if and only
if the leftmost and  rightmost leaves of the subtree of node $u$
partially overlaps with the range.
The 0th coefficient is assigned the weight $r$.
Let $O^{L}_{l}$ ($O^{R}_{l}$) be the size of the overlap sets for left (right)
subtree for $u$ with the range. Using reconstructed coefficients, we evaluate a query to produce answer $\widehat{R}$ as below.  
\[ 
\widehat{R}= r c_{0} + 2\sum_{l= 1}^{h} \Big( \frac{O^{L}_{l} - O^{R}_{l}}{2^{l}} \Big) \widehat{c}_{l}
\]
\noindent
where, $\widehat{c}_l$ is an unbiased estimation of a coefficient
$c_l$ at level $l$.
In the worst case, the absolute weight $|O^{L}_{l} -
O^{R}_{l}| = 2^{l}$.
We can analyze the corresponding varance, $V_{r}$, as 
\[
\begin{array}{rl}
V_{r} &\leq 2 \sum_{l=1}^{h} \Big(\frac{2^{l}}{2^{l+1}}\Big)^{2} \var[\widehat{c}_l] = \sum_{l=1}^{h}\frac{1}{2} \var[\widehat{c}_l] \\ 
&=  \frac{1}{2}\sum_{l=1}^{h} \frac{V_F}{p_l}
\end{array}
\]

Here, $V_F$ is the variance associated with the HRR frequency oracle.
As in the hierarchical case, the optimal choice is to set $p_l = 1/h$
(i.e. we sample a level uniformly), where $h = \log_2 (D)$.
Then we obtain
\begin{equation} \textstyle V_r = \frac12 \log^{2}_2(D) V_F
  \label{eq:haarvar}
  \end{equation}

It is instructive to compare this expression with the bounds obtained
for the hierarchical methods.
Recall that, after post-processing for consistency, we found that the
variance for answering range queries with HH$_8$, based on optimizing
the branching factor, is
$\log_2(r)\log_2(D) V_F/2$ (from \eqref{eq:hhvarbound}).
That is, for long range queries where $r$ is close to $D$, these
\eqref{eq:haarvar} will be close to \eqref{eq:hhvarbound}.
Consequently, we expect both methods to be competitive, and will use
empirical comparison to investigate their behavior in practice. 


Finally, observe that since this bound does not depend on the range
size itself, the average error across all possible range queries is
also bounded by \eqref{eq:haarvar}. 


\para{Key difference from the centralized case.}
The technique of perturbing Haar coefficients to answer differentially
private range queries was proposed and studied in the centralized case
under the name ``privelets''~\cite{dht:11}.
Subsequent work argued that more involved centralized algorithms could
obtain better accuracy. 
We will see in the experimental section that HaarHRR is among our best
performing methods.
Hence, our contribution in this work is to reintroduce the DHT as a
useful tool in local privacy. 

\subsection{Prefix and Quantile Queries}
\label{sec:quantiles}

Prefix queries form an important class of range queries, where the
start of the range is fixed to be the first point in the domain.
The methods we have developed allow prefix queries to be answered as a
special case.
Note that for hierarchical and DHT-based methods, we expect the error
to be lower than for arbitrary range queries.
Considering the error in hierarchical methods
(Theorem~\ref{thm:hhvar}), we require at most $B-1$ nodes at each
level to construct a prefix query, instead of $(2B-1)$, which reduces
the variance by almost half.
For DHT similarly, we only split nodes on the right end of a prefix
query, so we also reduce the variance bound by a factor of 2.
Note that a reduction in variance by 0.5 will translate into a factor
of $\sqrt{2} = 0.707$ in the absolute error.
Although the variance bound changes by a constant factor, we obtain
the same optimal choice for the branching factor in $B$.

Prefix queries are sufficient to answer quantile queries.
The $\phi$-quantile is that index $j$ in the domain such that at most
 a $\phi$-fraction of the input data lies below $j$, and at most a
$(1-\phi)$ fraction lies above it.
If we can pose arbitrary prefix queries, then we can binary search for
a prefix $j$ such that the prefix query on $j$ meets the
$\phi$-quantile condition.
Errors arise when the noise in answering prefix queries causes us
to select a $j$ that is either too large or too small.
The quantiles then describe the input data distribution in a general
purpose, non-parametric fashion. 
Our expectation is that our proposed methods should allow more
accurate reconstructions of quantiles than flat methods, since we
expect they will observe lower error. 
We formalize the problem: 

\begin{definition}(Quantile Query Release Problem)
 Given a set of $N$ users, the goal is to collect information
 guaranteeing $\epsilon$-LDP to approximate any quantile $q \in
 [0,1]$.
 Let $\widehat{Q}$ be the item returned as the answer to the quantile
 query $q$ using a mechanism $F$, which is in truth the $q'$ quantile, and
 let $Q$ be the true $q$ quantile.
We evaluate the quality of $F$  by both the {\em value error}, measured
by the squared error $(\widehat{Q}-Q)^{2}$; and the
{\em quantile error} $|q - \widehat{q}|$.
 \end{definition}

\eat{

\begin{figure*}
  \centering
\subfigure[Range Queries]{
   \includegraphics[scale=0.2]{figs/avg_err.png}\par
   
  \label{fig:mse_rq}
}

\subfigure[Prefix Queries]{
   \includegraphics[scale=0.2]{figs/avg_err_pre.png}\par
   
  \label{fig:mse_pq}
}
\caption{ The mean squared error incurred while answering all ${D \choose 2} $ queries and only  prefix queries as a function of $\epsilon$.}  
\label{fig:mse}
\end{figure*}

 }

 
\begin{figure*}[t]
    \includegraphics[width=\linewidth]{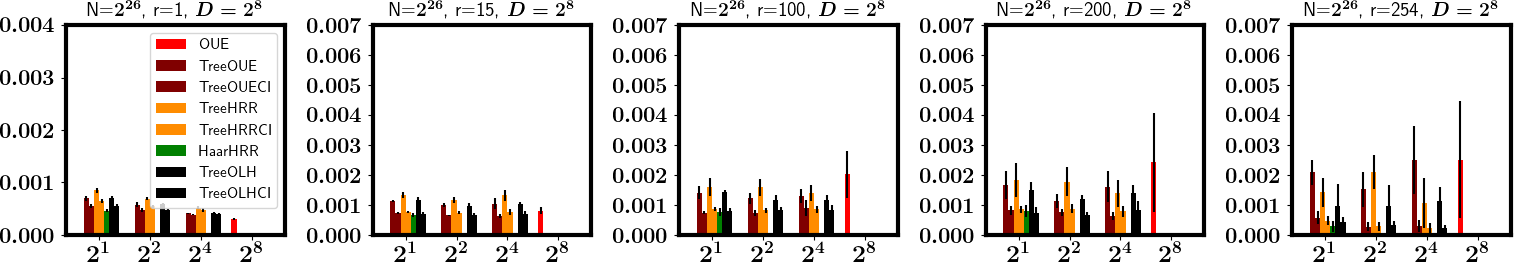}\par
    \includegraphics[width=\linewidth]{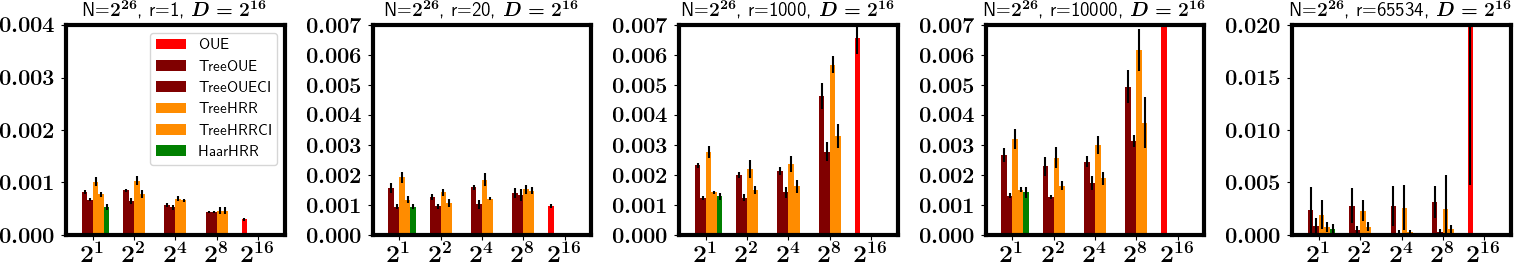}\par 
    \includegraphics[width=\linewidth]{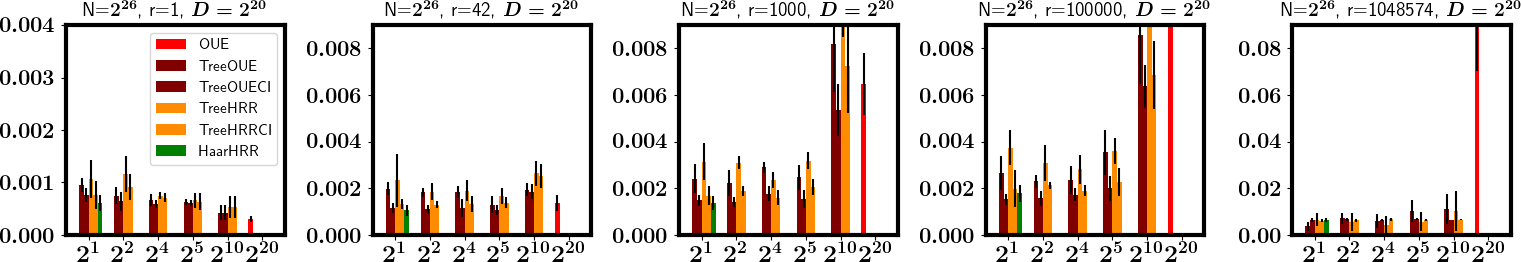}\par
    \includegraphics[width=\linewidth]{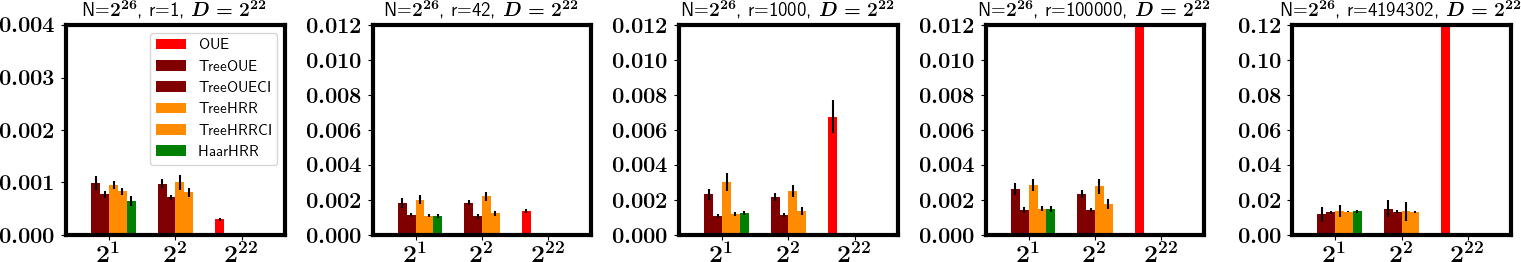}\par

\caption{Impact of constrained inference and branching factor $B$. In
  each plot, $B$ increases along the $X$ axis, and the $Y$ axis plots the mean squared error incurred in answering all range queries of length $r$.}
\label{fig:ci_b}
\end{figure*}

\section{Experimental Evaluation}
\label{sec:expts}
Our goal in this section is to validate our solutions and theoretical claims with 
experiments.

\para{Dataset Used.}
We are interested in comparing the flat, hierarchical and wavelet methods
for range queries of varying lengths on large domains,
capturing meaningful real-world settings. 
We have evaluated the methods over a variety of real and synthetic
data.
Our observation is that measures such as speed and accuracy do not
depend too heavily on the data distribution.
Hence, we present here results on synthetic data sampled from  Cauchy
distributions.
This allows us to easily vary parameters such as the population size
$N$ and the domain size $D$, as well as varying the distribution to be
more or less skewed. 

The shape of the (symmetrical) Cauchy distribution is
controlled by two parameters, center and height.
We set the location of the center at $P\times D$, for $0 < P < 1$, so that 
larger values of $P$ shift the mass further to the right.  
Since Cauchy distribution has infinite support, we drop any values
that fall outside $[D]$. 
Larger height parameters tend to reduce the sparsity in the
distribution by flattening it.
Our default choice is a relatively spread out distribution with height
= $\frac{D}{10}$ and $P=0.4$.\footnote{Our experiments varying $P$ did
  not show much significant deviation in the trends observed.}
We vary the domain size $D$ from small ($D=2^8$) to large
($D=2^{22}$) as powers of two. 


\para{Algorithm default parameters and settings.}
We set a default value of $e^{\epsilon}=3$ ($\epsilon=1.1$),
in line with prior work on LDP.
This means, for example, that binary randomized response will report a
true answer $\frac34$ of the time, and lie $\frac14$ of the time ---
enough to offer plausible deniability to users, while allowing
algorithms to achieve good accuracy. 
%
Since the domain size $D$ is chosen to be a power of 2, we can choose 
a range of branching factors $B$ for hierarchical histograms so that
$\log_B (D)$ remains an integer.
The default population size $N$ is set to be $N = 2^{26}$ which captures the scenario of an industrial deployment, similar to \cite{rappor:15,snap:18,applewhitepaper:17}.
Each bar plot is the mean of 5 repetitions of an experiment and error
bars capture the observed standard deviation.
The simulations are implemented in C++ and tested on a standard Linux
machine.
To the best of our knowledge, ours is among the first non-industrial
work to provide simulations with domain sizes as large as
$2^{22}$.
Our final implementation will shortly be made available as open source.

\para{Sampling range queries for evaluation.}
When the domain size is small or moderate ($D=2^{8}$ and $2^{16}$), it
is feasible to evaluate all ${D \choose 2 }$ range queries and so
compute the exact average. 
However, this is not scalable for larger domains, and so we average
over a subset of the range queries.
To ensure good coverage of different ranges, we pick a set of
evenly-spaced starting points, and then evaluate all ranges that begin
at each of these points. 
For $D=2^{20}$ and $2^{22}$ we pick start points every $2^{15}$ and
$2^{16}$ steps, respectively, yielding a total of 17M and 69M unique queries. 

\para{Histogram estimation primitives.}
The HH framework in general is agnostic to the choice of the histogram
estimation primitive $F$.
We show results with OUE, HRR and OLH as the primitives for histogram
reconstruction, 
since they are considered to be state of art~\cite{Wang:Blocki:Li:Jha:17}, and all provide the same theoretical bound
$V_F$ on variance.
Though any of these three methods can serve as a flat method, we
choose OUE as a flat method since it can be simulated efficiently and
reliably provides the lowest error in practice by a small margin. 
We refer to the hierarchical methods using HH framework as TreeOUE,
TreeOLH and TreeHRR.
Their counterparts where the aggregator applies postprocessing to
enforce consistency are identified with the CI suffix,
e.g. TreeHRRCI. 

We quickly observed in our preliminary experiments that direct
implementation of OUE can be very slow for large $D$:
the method perturbs and reports $D$ bits for each user. 
For accuracy evaluation purposes, we can replace the slow method with a
statistically equivalent simulation. 
That is, we can simulate the aggregated noisy count data that the
aggregator would receive from the population. 
We know that noisy count of any item is aggregated from two
distributions (1) ``true'' ones that are reported as ones (with
prob. $\frac{1}{2}$) (2) zeros that are flipped to be ones (with
prob. $\frac{1}{1+e^{\epsilon}}$).
Therefore, using the (private) knowledge of the true count $\theta[j]$
of item $j \in [D]$, the noisy count $\theta^{*}[j]$ can be expressed
as a sum of two binomial random variables, 
{$
\theta^{*}[j]=\text{Bino}(\theta[j],0.5) +
\text{Bino}\Big(N-\theta[j],\frac{1}{1+e^{\epsilon}}\Big).
$}
Our simulation can perform this sampling for all items, then provides the sampled
count to the aggregator, which then performs  the usual bias correction procedure.

The OLH method suffers from a more substantial drawback: the method is
very slow for the aggregator to decode, due to the need to iterate
through all possible inputs for each user report (time $O(ND)$).
We know of no short cuts here, and so we only consider OLH for our
initial experiments with small domain size $D$. 

\begin{figure*}[t]
\centering
\scalebox{0.8}{
\subfigure[$D=2^{8}$]{
\begin{tabular}{lrrrr}
\toprule
{$\epsilon$} &  HH$^{c}_{2}$ &  HH$^{c}_{4}$ &  HH$^{c}_{16}$ &   HaarHRR \\
\midrule
{0.2} &        4.269 &        4.037 &         4.176 &    \textbf{3.684} \\
{0.4} &        2.024 &        2.193 &         2.590 &    \textbf{1.831} \\
{0.6} &        1.388 &        1.341 &         1.535 &    \textbf{1.278} \\
0.8 &        1.002 &        \textbf{0.950} &         1.130 &    0.987 \\
1.0 &        0.844 &        \textbf{0.744} &         0.844 &    0.811 \\
1.1 &        0.722 &        \textbf{0.667} &         0.820 &    0.748 \\
1.2 &        0.684 &        0.658 &         \textbf{0.642} &    0.732 \\
1.4 &        0.571 &        \textbf{0.542} &         0.592 &    0.601 \\
\bottomrule
\end{tabular}
}
\subfigure[$D=2^{16}$]{
\begin{tabular}{lrrrr}
\toprule
{$\epsilon$} &  HH$^{c}_{2}$ &  HH$^{c}_{4}$ &  HH$^{c}_{16}$ &   HaarHRR \\
\midrule
{0.2} &   6.745 &        7.129 &         8.692 &    \textbf{6.666} \\
0.4 &        3.616 &        \textbf{3.424} &         4.648 &    3.526 \\
0.6 &        \textbf{2.333} &        2.360 &         2.793 &    2.342 \\
0.8 &        \textbf{1.644} &        1.728 &         2.075 &    1.711 \\
1.0 &        \textbf{1.356} &        1.377 &         1.642 &    1.484 \\
1.1 &        1.303 &        \textbf{1.270} &         1.597 &    1.345 \\
1.2 &        \textbf{1.090} &        1.140 &         1.433 &    1.201 \\
1.4 &        \textbf{0.922} &        0.995 &         1.158 &    1.130 \\
\bottomrule
\end{tabular}
}
%
%
\subfigure[$D=2^{20}$]{
\begin{tabular}{lrrrr}
\toprule
{$\epsilon$} &  HH$^{c}_{2}$ &  HH$^{c}_{4}$ &  HH$^{c}_{16}$ &   HaarHRR \\
\midrule
{0.2}&      10.043 &       10.493 &        11.511 &    \textbf{9.285} \\
{0.4}&       5.378 &        \textbf{4.751} &         5.617 &    5.261 \\
0.6&       3.605 &        \textbf{3.603} &         4.483 &    3.693 \\
0.8&       3.047 &        \textbf{3.042} &         3.352 &    3.316 \\
1.0&       \textbf{2.522} &        2.690 &         3.131 &    2.915 \\
1.1&       2.556 &        \textbf{2.540} &         2.729 &    2.722 \\
1.2&       2.619 &        \textbf{2.488} &         2.757 &    2.640 \\
1.4&       2.339 &        \textbf{2.304} &         2.652 &    2.505 \\      \bottomrule
\end{tabular}
}
\subfigure[$D=2^{22}$]{
\begin{tabular}{lrrrr}
\toprule
{$\epsilon$} &  HH$^{c}_{2}$ &  HH$^{c}_{4}$  &   HaarHRR \\
\midrule
{0.2} &        8.629 &        8.889 &    \textbf{8.422} \\
{0.4} &        4.546 &        4.951 &    \textbf{4.470} \\
{0.6} &        3.181 &        3.420 &    \textbf{3.085} \\
{0.8} &        2.657 &        2.692 &    \textbf{2.462} \\
1.0 &        \textbf{2.247} &        2.358 &    2.254 \\
1.1 &        \textbf{1.979} &        2.252 &    2.139 \\
{1.2} &        2.120 &        2.066 &    \textbf{1.946} \\
1.4 &        \textbf{1.650} &        1.885 &    1.990 \\
\bottomrule
\end{tabular}
}
}
\caption{Impact of varying $\epsilon$ on mean squared error for arbitrary queries. These numbers are scaled up by 1000 for presentation.}
\label{tab:mse_range}
\centering
\scalebox{0.8}{

\subfigure[$D=2^{8}$]{
\begin{tabular}{lrrrr}
\toprule
{$\epsilon$} &  HH$^{c}_{2}$ &  HH$^{c}_{4}$ &  HH$^{c}_{16}$ &   HaarHRR \\
\midrule
{0.2} &        4.306 &        \underline{2.968} &         4.282 &    \underline{\textbf{2.857}} \\
{0.4} &        \underline{1.859} &        \underline{1.439} &         \underline{1.828} &    \underline{\textbf{1.377}} \\
0.6 &        \underline{1.366} &        \underline{\textbf{0.957}} &   \underline{1.758} &    \underline{1.031} \\
{0.8} &        \underline{0.937} &        \underline{0.778} &         \underline{0.896} &    \underline{\textbf{0.758}} \\
1.0 &        \underline{0.802} &        \underline{\textbf{0.561}} & \underline{0.637} &    \underline{0.613} \\
1.1 &        \underline{0.684} &   \underline{\textbf{0.533}} &   \underline{0.666} &    \underline{0.626} \\
1.2 &        \underline{0.658} &  \underline{\textbf{0.437}} &         0.670 &    \underline{0.568} \\
1.4 &        0.573 &        \underline{\textbf{0.420}} &         \underline{0.478} &    \underline{0.494} \\
\bottomrule
\end{tabular}
}
\subfigure[$D=2^{16}$]{
\begin{tabular}{lrrrr}
\toprule
{$\epsilon$}&  HH$^{c}_{2}$ &  HH$^{c}_{4}$ &  HH$^{c}_{16}$ &   HaarHRR \\
\midrule
{0.2} &        7.701 &        \underline{6.172} &  \underline{7.014} &    \underline{\textbf{5.870}} \\
{0.4} &        \underline{3.266} &     \underline{3.101} &         \underline{3.744} &    \underline{\textbf{2.880}} \\
{0.6} &        2.402 &        \underline{2.176} &  \underline{2.426} &    \underline{\textbf{2.018}} \\
0.8 &        1.663 &        \underline{\textbf{1.503}} &  \underline{1.834} &    \underline{1.511} \\
1.0 &        \underline{1.338} &   \underline{\textbf{1.220}} &  \underline{1.426} &    \underline{1.244} \\
1.1 &        \underline{1.202} &        \underline{\textbf{1.051}} &   \underline{1.259} &    \underline{1.120} \\
1.2 &        \underline{1.080} &  \underline{\textbf{0.978}} &     \underline{1.147} &    \underline{1.054} \\
1.4 &        0.973 &        \underline{\textbf{0.848}} &   \underline{0.981} &    \underline{0.973} \\
\bottomrule
\end{tabular}
}
%
%
\subfigure[$D=2^{20}$]{
\begin{tabular}{lrrrr}
\toprule
{$\epsilon$}&  HH$^{c}_{2}$ &  HH$^{c}_{4}$ &  HH$^{c}_{16}$ &   HaarHRR \\
\midrule
{0.2} &       \underline{ 8.874} &        \underline{8.255} &        \underline{10.462} &    \underline{\textbf{7.237}} \\
{0.4} &        \underline{4.734} &        \underline{4.395} &         5.754 &    \underline{\textbf{4.271}} \\
{0.6} &        3.788 &        \underline{3.485} &         \underline{4.055} &    \underline{\textbf{3.377}} \\
0.8 &        3.287 &        \textbf{3.094} &         \underline{3.268} &   \underline{3.108} \\
1.0 &        3.022 &        2.848 &         \underline{\textbf{2.826}} &    2.920 \\
{1.1} &        3.053 &        2.756 &         \underline{\textbf{2.727}} &    \textbf{2.727} \\
1.2 &        3.145 &        \textbf{2.627} &         2.914 &    2.754 \\
1.4 &        2.975 &        2.659 &         \underline{\textbf{2.543}} &    2.696 \\
      \bottomrule
\end{tabular}
}
\subfigure[$D=2^{22}$]{
\begin{tabular}{lrrrr}
\toprule
{$\epsilon$}&  HH$^{c}_{2}$ &  HH$^{c}_{4}$  &   HaarHRR \\
\midrule
{0.2} &        \underline{8.620} &        \underline{8.638} &    \underline{\textbf{8.099}} \\
0.4 &        		  \underline{\textbf{4.181}} &        \underline{4.330} &    \underline{4.233} \\
0.6 &        		  \underline{\textbf{2.932}} &        \underline{3.077} &    \underline{3.063} \\
0.8 &        		  \underline{\textbf{2.215}} &        \underline{2.590} &    \underline{2.528} \\
1.0 &        		  \underline{\textbf{1.958}} &        \underline{2.246} &    2.326 \\
1.1 &        		  \underline{\textbf{1.777}} &        2.319 &    2.181 \\
1.2 &        		  \underline{\textbf{1.929}} &        2.174 &    2.205 \\
1.4 &        		  \underline{\textbf{1.613}} &        \underline{1.868} &    2.156 \\
      \bottomrule
\end{tabular}
}
}
\caption{Impact of varying $\epsilon$ on mean squared for prefix queries. These numbers are scaled up by 1000 for presentation. We underline the scores that are smaller than corresponding scores in Figure~\ref{tab:mse_range}. }
\label{tab:mse_prefix}

\end{figure*}

\begin{figure}[t]
\centering
\begin{tabular}{|c|c|c|c|c|}
\hline 
D & $2^{8}$ & $2^{9}$ & $2^{10}$ & $2^{11}$ \\ 
\hline 
Wavelet & 221.62 & 306.31 & 410.29 & 536.32 \\ 
(optimal) HH$^{c}_{16}$ & 79.23 & 164.48 & 185.94 & 213.87 \\ 
HH$^{c}_{2}$ & 220.06 & 305.54 & 409.48 & 535.63 \\ 
\hline 
$\frac{\text{Wavelet}}{ \text{HH}^{c}_{16}}$ & 2.7971 & 1.8622 & 2.20 & 2.5077 \\ 
$\frac{\text{HH}^{c}_{2}}{ \text{HH}^{c}_{16}}$ & 2.777 & 1.8576 & 2.202 & 2.5044 \\ 
\hline 
\end{tabular}

\caption{Table 3 from \protect\cite{QYL:13} comparing the exact average variance incurred in answering all range queries for $\epsilon=1$ in the centralized case.} 
\label{tab:haar_vs_hh}
\end{figure}

\subsection{Impact of varying $B$ and $r$} 

\para{Experiment description.}
In this experiment, we aim to study how much a privately reconstructed
answer for a range query deviates from the ground truth.
 Each query answer is normalized to fall in the range 0 to 1, so we
 expect good results to be much smaller than 1. 
To compare with our theoretical analysis of variance, we measure the accuracy in the form of mean squared error between true and reconstructed range query answers.

\para{Plot description.}
Figure~\ref{fig:ci_b} illustrates the effect of branching factor $B$ 
on accuracy for domains of size $2^{8}$ (small), $2^{16}$  (medium),
and lastly $2^{20}$ and $2^{22}$ (large).
Within each plot with a fixed $D$ and query length $r$, we vary the branching factor
on the $X$ axis.
We plot the flat OUE method as if it were a hierarchical method with
$B=D$, since it effectively has this fan out from the root. 
We treat HaarHRR as if it has $B=2$, since is based on a binary tree
decomposition. 
The $Y$ axis in each plot shows the mean squared error incurred while
answering all queries of length $r$.
As the plots go left to right, the range length increases from 1 to
just less than the whole domain size $D$.  
The top row of plots have $D=2^8$, and the last row of plots have
$D=2^{22}$. 

\para{Observations.}
Our first observation is that
the CI step reliably provides a significant improvement in accuracy in
almost all cases for HH, and never increases the error.
Our theory suggests that the CI step improves the worst case accuracy
by a constant factor, and this is borne out in practice.
This improvement is more pronounced at larger intervals and higher
branching factors.
In many cases, especially in the right three columns, TreeOUECI and TreeHRRCI are two to four times more accurate then their inconsistent counter parts.   
Consequently, we put our main focus on methods with consistency
applied in what follows. 

Next, we quickly see evidence that the flat
approach (represented by OUE) is not effective for answering range queries.  
Unsurprisingly, for point queries ($r=1$), flat methods are
competitive.
This is because all methods need to track information on individual
item frequencies, in order to answer short range queries.
The flat approach keeps only this information, and so maximizes the
accuracy here.
Meanwhile, HH methods only use leaf level information to answer point queries, and so we see
better accuracy the shallower the tree is, i.e. the bigger $B$ is. 
However, as soon as the range goes beyond a small fraction of the
domain size, other approaches are preferable.
The second column of plots shows results for relatively short ranges
where the flat method is not the most accurate. 


For larger domain sizes and queries, our methods outperform the flat
method by a high margin.
For example, the best hierarchical methods for very long queries and
large domains are at least 16 times more accurate than the
flat method.
Recall our discussion of OLH above that emphasised that its
computation cost scales poorly with domain size $D$. 
We show results for TreeOLH and TreeOLHCI for the small domain size
$2^{8}$, but drop them for larger domain sizes, due to this poor
scaling.
We can observe that although the method acheives competitive accuracy,
it is equalled or beaten by other more performant methods, so we are
secure in omitting it. 

As we consider the two tree methods, TreeOUE and TreeHRR, we observe
that they have similar patterns of behavior.
In terms of the branching factor $B$, 
it is difficult to pick a single particular $B$ to minimize the
variance, due to the small relative differences.
The error seems to decrease from $B=2$, and increase for larger $B$
values above $2^4$ (i.e. 16). 
Across these experiments, we observe that choosing $B = 4$, 8 or 16
consistently provides the best results for medium to large sized
ranges.
This agrees with our theory, which led us to favor $B=8$ or $B=4$,
with or without consistency applied respectively. 
This range of choices means that  we are not penalized severely for failing to choose an optimal value of $B$.

The main takeaway from Figure~\ref{fig:ci_b} is the strong performance
for the HaarHRR method.  
It is not competitive for point queries ($r=1$), but for all ranges
except the shortest it achieves the single best or equal best
accuracy.
For some of the long range queries covering the almost the entire
domain, it is slightly outperformed by consistent HH$_B$ methods.
However, this is sufficiently small that it is hard to observe
visually on the plots. 
Across a broad range of query lengths (roughly, 0.1\% to 10\% of the
domain size), HaarHRR is preferred.
It is most clearly the preferred method for smaller domain sizes, such
as in the case of $D=2^8$.
We observed a similar behavior for domains as small as $2^{5}$.


\subsection{Impact of privacy parameter $\epsilon$}

\para{Experiment description.}
We now vary $\epsilon$ between 0.1 (higher privacy) to 1.4 (lower
privacy) and find the mean squared error over range queries.
Similar ranges of $\epsilon$ parameters are used in prior works such
as \cite{calm:2018}.
After the initial exploration documented in the previous section, our
goal now is to focus in on the most accurate and scalable hierarchical methods.
Therefore, we omit all flat methods and consider only those values of
$B$ that provided satisfactory accuracy.
We choose TreeOUECI as our mechanism to instantiate HH (henceforth
denoted by HH$^{c}_B$, where the $c$ denotes that consistency is applied)
method due to its accuracy. 
We do note that a deployment may prefer TreeHRRCI over TreeOUECI since
it requires vastly reduced communication for each user at the cost of
only a slight increase in error.

\para{Plot description.}
Table~\ref{tab:mse_range} compares the mean squared error for
HH$^{c}_2$, HH$^{c}_4$ HH$^{c}_{16}$ and HaarHRR for various
$\epsilon$ values.
We multiply all results by a factor of 1000 for convenience, so the
typical values are around $10^{-3}$
corresponding to very low absolute error.
In each row, we mark in bold the lowest observed variance, noting that
in many cases, the ``runner-up'' is very close behind. 

\para{Observations.}
The first observation, consistent with Figure~\ref{fig:ci_b}, is that
for lower $\epsilon$'s, HaarHRR is more accurate than the best of
HH$^{c}_B$ methods.
This improvement is most pronounced for $D=2^{8}$ i.e. at most 10\%
(at $\epsilon=0.2$) and marginal (0.01 to 1\%) for larger domains.
For larger $\epsilon$ regimes, HH$^{c}_B$ 
outperforms HaarHRR, but only by a small margin of at most 11\%.
For large domains, HH$^{c}_B$ remains the best method.
In general, except for $D=2^{22}$, there is no one value of $B$ that achieves the best
results at all parameters but overall $B=4$ yields slightly more
accurate results for HH$^{c}_B$ for most cases.
Note that this $B$ value is closer to the optimal value of 9 (derived
in Section~\ref{sec:consistency}) than other values.
When $D=2^{22}$, HH$_{2}^{c}$ dominates HH$_{4}^{c}$ but only by a margin of at most 10\%. 

\para{Comparison with DHT and HH based approaches in the centralized case.}
We briefly contrast with the conclusion in the centralized case.
We reproduce some of the results of Qardaji et al.~\cite{QYL:13} in Table~\ref{tab:haar_vs_hh},
comparing variance for the (centralized) wavelet based approach to (centralized) hierarchical histogram approaches with $B=2,16$ with
consistency applied.
These numbers are scaled and not normalized, so can't be directly compared to our results (although, we know that the error should be much lower in the
centralized case).
However, we can meaningfully compare the ratio of variances, which we show
in the last two rows of the table. 
\eat{
compare ~\cite{dht:11}'s haar
wavelet approach with HH$^{c}_B$ and other methods based on the exact
average error incurred in answering all queries.
We reproduce relevant rows from this table in
table~\ref{tab:haar_vs_hh}.}

\par For $\epsilon=1,D=2^{8}$, the error for the Haar method is
approximately $2.8$ times more than the hierarchical approach. 
Meanwhile, the corresponding readings for HaarHRR and HH$^{c}_4$ (the
most accurate method in the $\epsilon=1$ row) in Table~\ref{tab:mse_range} are
0.787 and 0.763 --- a deviation of only $\approx$3\%. Another
important distinction from the centralized case is that we are not
penalized a lot for choosing a sub-optimal branching factor. 
Whereas, 
we see in the 4th row that choosing $B=2$ increases the error of consistent HH method by at least 1.8576 times from the preferred method HH$^{c}_{16}$.  

A further observation is that (apart for $D=2^{22}$)
across 24 observations, HaarHRR is
never outperformed by \emph{all} values of HH$^{c}_B$ i.e. in no
instance is it the least accurate method.
It trails the best HH$^{c}_B$ method by at most 10\%.
On the other hand, in the centralized case
(Table~\ref{tab:haar_vs_hh}), the variance for the wavelet based approach is at least 1.86 times higher than HH$^{c}_B$.


\eat{
\begin{figure*}
    \includegraphics[width=\linewidth]{figs/d_8_pre.png}
    \includegraphics[width=\linewidth]{figs/d_16_pre.png} 
     \includegraphics[width=\linewidth]{figs/d_18_pre.png}
\caption{Impact of constrained inference on HH framework on prefix queries. We increase $r$ along X axis and Y axis plots the mean squared error incurred in answering all range queries of length $r$.}
\label{fig:dd3}
\end{figure*}
}

\begin{figure*}[t]
  \centering
    \includegraphics[width=0.9\linewidth]{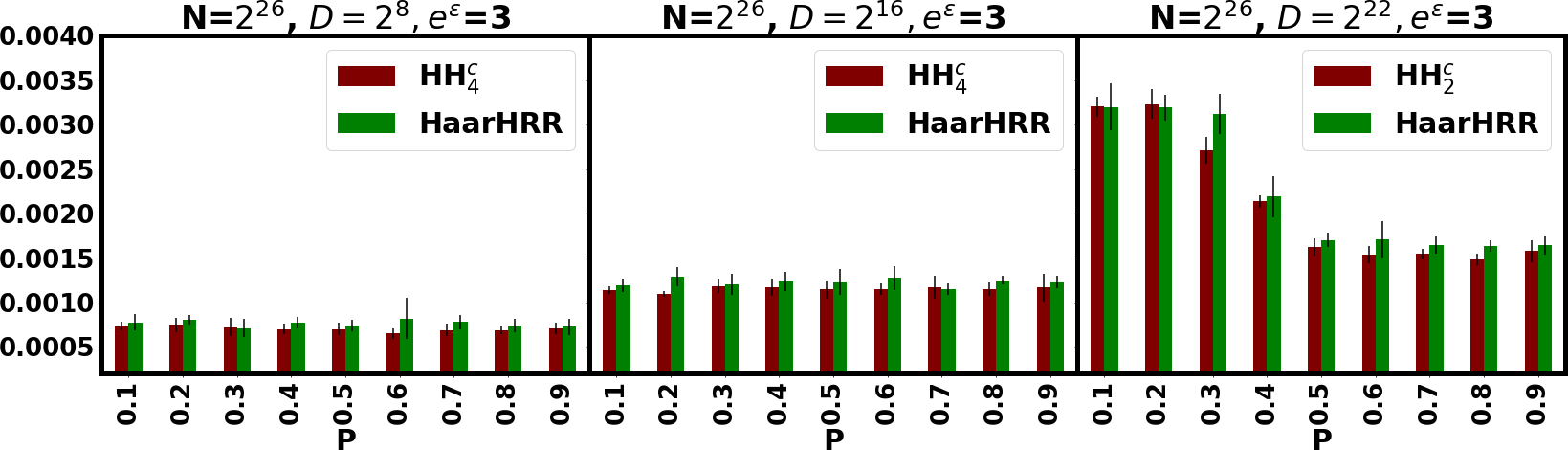}
\caption{Impact of varying the distribution center ($P \times D$) on
  mean squared error for various domain sizes $D$.
}
\label{fig:mse_dbn}
\end{figure*}

\subsection{Prefix Queries}

\para{Experiment description.}
As described in Section~\ref{sec:quantiles}, prefix queries deserve special
attention.
Our set up is the same as for range queries.
We evaluate every prefix query, as there are fewer of them. 

\para{Plot description.} Table~\ref{tab:mse_prefix} is the analogue of
Table~\ref{tab:mse_range} for prefix queries, computed with the same
settings.
We underline the scores that are smaller than corresponding scores in Table~\ref{tab:mse_range}. 

\para{Observations.} The first observation is that the error in
Table~\ref{tab:mse_prefix} is often smaller (up to 30\%) than in
Table~\ref{tab:mse_range} at many instances, particularly for small and
medium sized domains.
The reduction is not as sharp as the analysis might suggest,
since that only gives upper bounds on the variance. 
Reductions in error are not as noticeable for larger values of $D$,
although this could be impacted by our range query sampling strategy. 
In terms of which method is preferred, 
HH$^{c}_{2}$ for $D=2^{22}$ and HH$_{4}^{c}$ tend to dominate for
larger $\epsilon$, while  HaarHRR is preferred for smaller $\epsilon$.

\subsection{Impact of input distribution}

\para{Experiment description.} We now check whether the shape of the input distribution affects the mean squared error when other parameters are held to their default values. 

\para{Plot description.}
Figure~\ref{fig:mse_dbn} plots the mean squared error for domains of
different sizes for $e^\epsilon=3$ ($\epsilon=1.1$).
Along the $X$ axis, we shift the center of the Cauchy distribution by changing  $0 <P <1$.
For each domain size $D$, we make
our comparison between HaarHRR and the most accurate consistent HH method according to Table~\ref{tab:mse_range}. 

\para{Observations.}
The chief observation is that for small and mid sized domains, the
change in distribution does not make any noticeable difference in
the accuracy.
HaarHRR continues to be slightly inferior to
HH$_4^{c}$ for all input shapes.
For $D=2^{22}$, we do notice an increase in the error when the bulk of
the mass of the distribution is towards the left end of the domain
($P=0.1$ to $P=0.3$).
This is partly due to the range sampling method we use: the
majority of range queries we test cover only a small amount of the
true probability mass of the distribution for these $P$ values, and
this leads to increased error from privacy noise.
However, the main take-away here should be the consistently small
absolute numbers: maximum mean squared error of 0.0035, i.e. very accurate answers.

\begin{figure}[t]
\centering
  \subfigure[$P=0.1$]{
    \includegraphics[width=0.5 \linewidth]{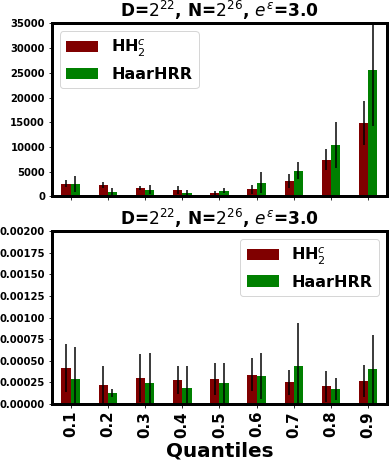}
    }%
    \subfigure[$P=0.5$]{
    \includegraphics[width=0.5 \linewidth]{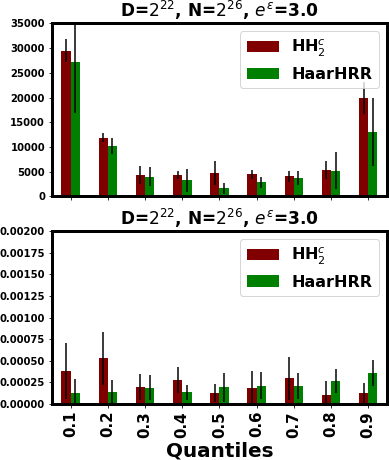} 
    }
    \caption{
Top row: value error; bottom row: quantile error}   
\label{fig:quantiles}
\end{figure}

\subsection{Quantile Queries}
\para{Experiment description.}
Finally, we compare the performance of the best hierarchical
approaches in evaluation of the deciles
(i.e. the $\phi$-quantiles for $\phi$ in 0.1 to 0.9)
for a left skewed ($P=0.1$) and centered ($P=0.5$) Cauchy distribution.

\para{Plot description.}
The top row in Figure~\ref{fig:quantiles} plots the actual difference
between true and reconstructed quantile values (value error).
The corresponding bottom plots measure the absolute difference between
the quantile value of the returned value and the target quantile
(quantile error). 

\para{Observations.}
The first observation is that the both the algorithms have low
absolute value error (the top row).
For the domain of $2^{22} \approx 4M$, 
even the largest error of $\approx$35K made by HH$_{2}^{c}$ is still
very small, and less than 1\%. 
The value error tends to be highest where the data is least dense:
towards the right end when the data skews left ($P=0.1$), and at
either extreme when the data is centered. 
Importantly, the corresponding quantile error is mostly flat.
This means that instead of finding the median (say), our methods
return a value that corresponds to the 0.5004 quantile, which is very
close in the distributional sense.  
This reassures us that any spikes in the value error are mostly a
function of sparse data, rather than problems with the methods. 


\subsection{Experimental Summary}

We can draw a number of conclusions and recommendations from this
study:

\para{$\bullet$}
The flat methods are never competitive, except for very short ranges and
small domains.


\para{$\bullet$}
  The wavelet approach is preferred for small values of $\epsilon$ 
  (roughly $\epsilon < 0.8$), while the (consistent) HH approach is preferred
  for larger $\epsilon$'s and for larger queries.

  \para{$\bullet$}
This threshold is slightly reduced for larger domains.
  However, the ``regret'' for choosing a ``wrong'' method is low:
  the difference between the best method and its competitor from HH and
  wavelet is typically no more than 10-\%. 

\para{$\bullet$}
Overall, the wavelet approach (HaarHRR) is always a good compromise
method. 
It provides accuracy comparable to consistent HH in all settings, and
requires a constant factor less space ($D$ wavelet coefficients
against $2D-1$ for HH$_2$). 

%
%


\section{Concluding Remarks}
\label{sec:conclusions}
\label{sec:extensions}

We have seen that we can accurately answer range queries under the
model of local differential privacy.
Two methods whose counterparts have quite differing behavior in the
centralized setting are very similar under the local setting, in line
with our theoretical analysis.
Now that we have reliable primitives for range queries and quantiles,
it is natural to consider how to extend and apply them further.

\para{Multidimensional range queries.}
Both the hierarchical and wavelet approaches can be extended to
multiple dimensions.
Consider applying the hierarchical decomposition to two-dimensional
data, drawn from the domain $[D]^2$.
Now any (rectangular) range can be decomposed into
$\log_B^2 D$ $B$-adic rectangles (where each side is drawn from a
$B$-adic decomposition), and so we can bound the variance in terms of 
$\log_B^4 D$.
More generally, we achieve variance depending on $\log^{2d} D$ for
$d$-dimensional data.
Similar bounds apply for generalizations of wavelets.
These give reasonable bounds for small values of $d$ (say, 2 or 3). 
For higher dimensions, we anticipate that coarser gridding approaches
would be preferred, in line with~\cite{Qardaji:Yang:Li:12}. 

\para{Advanced data analysis.}
In the abstract, many tasks in data modeling and prediction can be
understood as building a description of observed data density.
For example, many (binary) classification problems can be described as
trying to predict what class is most prevalent
in the neighborhood of a given query point. 
Viewed through this lens, range queries form a primitive that can be
used to build such model.
As a simple example, consider building a Naive Bayes classifier for a
public class based on private numerical attributes.
If we use our methods to allow range queries to be evaluated on each
attribute for each class, we can then build models for
the prediction problem.
Generalizations of this approach to more complex models, different
mixes of public and private attributes, and different questions,
give a  set of open problems for this area. 


\eat{
\begin{proof}[ of Lemma~\ref{lemma:prefixcoefs}]
 From the left end (i.e. 0), $h- \lceil \log_2 (r) \rceil$ have non-zero contributions in the answer. These nodes to be included are the ones in the path from the lowest ancestor $v$ of 0 and b to the root of the tree. Similarly, the  nodes that participate in the evaluation from b's end are the ones in the path from $b$ to $v$. There are $\lceil \log_2 (r) \rceil$ of these. Therefore, at most total $h-\lceil \log_2 (r) \rceil + \lceil \log_2 (r) \rceil=h$ nodes are involved in the answer.
\end{proof}
}

\para{Acknowledgements.}
This work is supported in part by
AT\&T, European Research Council grant ERC-2014-CoG 647557, 
and The Alan Turing Institute under the EPSRC grant EP/N510129/1.

\bibliographystyle{ACM-Reference-Format}
\bibliography{papers}

\end{document}